%% file: main.tex
\documentclass[11pt]{article}
\input{macros}

\title{Coloring in Graph Streams via Deterministic and\\ Adversarially Robust Algorithms\footnote{
This work was supported in part by NSF under awards CCF-1907738 and CCF-2006589.
}}
\author{Sepehr Assadi\thanks{Department of Computer Science, Rutgers University. Research supported in part by a
NSF CAREER Grant CCF-2047061, a Google Research gift, and a Fulcrum award from Rutgers Research Council.} \and Amit Chakrabarti\thanks{Department of Computer Science, Dartmouth College.} \and Prantar Ghosh\thanks{DIMACS, Rutgers Univeristy. Work done in part while the author was at Dartmouth College.} \and Manuel Stoeckl\footnotemark[3]}
\date{}

\begin{document}

\maketitle
\thispagestyle{empty}
\input{abstract}

\newpage

\addtocounter{page}{-1}
\input{intro-v2}

\input{prelim}

\input{determ}

\input{robust}

\bibliographystyle{alpha}
\bibliography{refs}

\appendix
\input{tech-proofs}

\end{document}

%% file: macros.tex
\usepackage[margin=1in]{geometry}
\usepackage[utf8]{inputenc}
\usepackage[usenames, dvipsnames]{xcolor}
\usepackage{amsmath, amsthm, amssymb, thmtools, calc}
\usepackage{fourier, bm, bbm}  
\usepackage[T1]{fontenc}
\usepackage{mleftright, xspace, paralist, enumitem, multirow, booktabs}
\usepackage{algorithm}
\usepackage[noend]{algpseudocode}
\usepackage{url}
\usepackage[colorlinks, linkcolor=BrickRed, citecolor=blue]{hyperref}
\usepackage{cite, cleveref}

\makeatletter
\def\th@plain{%
  \thm@notefont{}
  \itshape 
}
\def\th@definition{%
  \thm@notefont{}
  \normalfont 
}
\makeatother
\newtheorem{theorem}{Theorem}
\newtheorem{lemma}{Lemma}[section]
\newtheorem{corollary}[lemma]{Corollary}

\theoremstyle{definition}  \newtheorem{definition}[lemma]{Definition}

\theoremstyle{remark}  
\newlist{thmparts}{enumerate}{1}
\setlist[thmparts]{labelindent=\parindent,leftmargin=*,itemsep=2pt,font=\normalfont,label=(\thetheorem.\arabic*)}
\Crefname{thmpartsi}{Subresult}{Subresults} 
\declaretheoremstyle[%
  spaceabove=-6pt,%
  spacebelow=6pt,%
  headfont=\normalfont\itshape,%
  postheadspace=1em,%
  qed=\qedsymbol%
]{mystyle} 
\declaretheorem[name={Proof},style=mystyle,unnumbered,
]{prf}
\RequirePackage{caption, float}

\algrenewcomment[1]{\hfill \textcolor{blue}{$\triangleright$ #1}}
\algnewcommand{\LineComment}[1]{\State \textcolor{blue}{$\triangleright$ #1}}
\algdef{SE}[SUBALG]{Indent}{EndIndent}{}{\algorithmicend\ }%
\algtext*{Indent}
\algtext*{EndIndent}

\DeclareMathOperator{\EE}{\mathbb{E}}
\newcommand{\FF}{\mathbb{F}}

\newcommand{\ZZ}{\mathbb{Z}}
\newcommand{\indic}{\mathbbm{1}}

\newcommand{\cC}{\mathcal{C}}

\newcommand{\cF}{\mathcal{F}}

\newcommand{\cH}{\mathcal{H}}

\newcommand{\cP}{\mathcal{P}}
\newcommand{\cQ}{\mathcal{Q}}
\newcommand{\cR}{\mathcal{R}}

\newcommand{\cU}{\mathcal{U}}

\newcommand{\bzero}{\mathbf{0}}
\newcommand{\bone}{\mathbf{1}}
\newcommand{\ba}{\mathbf{a}}

\newcommand{\bi}{\mathbf{i}}
\newcommand{\bj}{\mathbf{j}}

\newcommand{\bw}{\mathbf{w}}

\newcommand{\tO}{\widetilde{O}}

\DeclareMathOperator{\Free}{Free}
\DeclareMathOperator{\mono}{mono}
\DeclareMathOperator{\poly}{poly}
\DeclareMathOperator{\polylog}{polylog}

\DeclareMathOperator{\Var}{Var}
\DeclareMathOperator{\Nhd}{N}

\newcommand{\slackwrt}[2]{\slack(#2 \mid #1)}  
\newcommand{\invalpha}{\alpha^{-1}}            
\newcommand{\dconf}{\mathrm{d}_{\textup{conf}}}  

\newcommand{\ang}[1]{{\langle{#1}\rangle}}
\newcommand{\ceil}[1]{{\left\lceil{#1}\right\rceil}}
\newcommand{\floor}[1]{{\left\lfloor{#1}\right\rfloor}}
\renewcommand{\b}{\{0,1\}}

\newcommand{\setm}{\smallsetminus}

\newcommand{\ontop}[2]{\genfrac{}{}{0pt}{2}{#1}{#2}}
\newcommand{\expln}[2]{\stackrel{\text{\scriptsize{#1}}}{#2}}   

\makeatletter
\newcommand\squeezepar{\@startsection{paragraph}{4}{\z@}{1.5ex \@plus1ex \@minus.2ex}{-1em}{\normalfont\normalsize\bfseries}}
\allowdisplaybreaks
\renewcommand{\paragraph}[1]{\squeezepar{{#1}.}}
\makeatother

\newcommand{\curr}{\mathtt{curr}}

%% file: abstract.tex
\begin{abstract}

\noindent 
In recent years, there has been a growing interest in solving various graph
coloring problems in the streaming model. The initial algorithms in this line
of work are all crucially randomized, raising natural questions about how
important a role randomization plays in streaming graph coloring. A couple of
very recent works have made progress on this question: they prove that
deterministic or even adversarially robust coloring algorithms (that work on
streams whose updates may depend on the algorithm's past outputs) are
considerably weaker than standard randomized ones. However, there is still a
significant gap between the upper and lower bounds for the number of colors
needed (as a function of the maximum degree $\Delta$) for robust coloring and multipass
deterministic coloring. We contribute to this line of work by proving the
following results. 
  
  \begin{itemize}
    \item In the deterministic semi-streaming (i.e., $O(n \cdot
    \text{polylog } n)$ space) regime, we present an algorithm that achieves a
    combinatorially optimal $(\Delta+1)$-coloring using
    $O(\log{\Delta} \log\log{\Delta})$ passes.  This improves upon the
    prior $O(\Delta)$-coloring algorithm of Assadi, Chen, and Sun (STOC 2022)
    at the cost of only an $O(\log\log{\Delta})$ factor in the number of
    passes.   
    \item In the adversarially robust semi-streaming regime, we design an
    $O(\Delta^{5/2})$-coloring algorithm that improves upon the previously
    best $O(\Delta^{3})$-coloring algorithm of Chakrabarti, Ghosh, and Stoeckl
    (ITCS 2022). Further, we obtain a smooth colors/space tradeoff that
    improves upon another algorithm of the said work: whereas their algorithm
    uses $O(\Delta^2)$ colors and $O(n\Delta^{1/2})$ space, ours, in
    particular, achieves (i)~$O(\Delta^2)$ colors in $O(n\Delta^{1/3})$ space,
    and (ii)~$O(\Delta^{7/4})$ colors in $O(n\Delta^{1/2})$ space.   
  \end{itemize}

\end{abstract}

%% file: intro-v2.tex
\section{Introduction}

In the graph coloring problem, we are given an undirected graph and the goal
is to assign colors to the nodes of the graph such that adjacent nodes receive
different colors.  Graph coloring is a fundamental problem in graph theory
with numerous applications in computer science, including in databases, data
mining, register allocation, and
scheduling~\cite{Chaitin82,LotfiS86,peng2016vcolor}; see, e.g., the
application to parallel query optimization by Hasan and
Motwani~\cite{HasanM95}. The emergence of massive graphs in many
of these application domains has necessitated the study of graph coloring
algorithms that are capable of handling such graphs efficiently on modern
architecture. Of particular interest is the family of \emph{graph streaming}
algorithms: each such algorithm computes its solution using only a small
number of sequential passes over the edges of the input graph, while using a
sublinear amount of memory. 

Several graph coloring problems have been studied in the streaming setting,
typically with the goal of achieving a palette size (total number of colors
used) proportional to the graph's chromatic
number~\cite{CormodeDK19,AbboudCKP19}, maximum vertex-degree~\cite{AssadiCK19,
BeraG18, AlonA20, AssadiCS22, AssadiKM22}, arboricity~\cite{BeraG18}, or
degeneracy~\cite{BeraCG20}. Also studied is the closely-related problem of
(degree+1)-list-coloring~\cite{HalldorssonKNT22} (see also~\cite{AlonA20}).
Furthermore, graph coloring has been considered under different streaming
paradigms such as random stream order and the vertex-arrival model
\cite{BhattacharyaBMU21}. Most of these works consider the
\emph{semi-streaming} regime, where we are restricted to $O(n\cdot
\text{polylog } n)$ space for processing an $n$-vertex graph.
Since even just storing the output coloring can require $\Omega(n\log n)$ space,
this is close to optimal for the problem. We study
semi-streaming graph coloring, focusing on  \emph{the} most popular color
parameter in this line of work, namely the maximum degree $\Delta$ of the
graph: we call this ``$\Delta$-based coloring.''

A trivial greedy algorithm achieves a $(\Delta+1)$-coloring in the offline
setting. However, obtaining this color bound in the \emph{streaming} model is
fairly challenging. A breakthrough work by Assadi, Chen, and Khanna
\cite{AssadiCK19} did achieve such a coloring in semi-streaming space. An
aspect of this algorithm, shared with almost all subsequent streaming coloring
algorithms, is that it is inherently randomized. This raises the natural
question: to what extent is \emph{randomization} necessary for $\Delta$-based
coloring? 
Indeed, a derandomized version can be advantageous in multiple scenarios, e.g., having low or zero error even when the algorithm is rerun a huge (maybe exponential) number of times, or for robustness against input streams generated based on the algorithm's past outputs or internal states. 

Two recent works have addressed this question. On the one hand, Assadi, Chen,
and Sun~\cite{AssadiCS22} ruled out non-trivial single-pass deterministic
algorithms for $\Delta$-based coloring: any such algorithm requires
$\exp(\Delta^{\Omega(1)})$ colors for semi-streaming space (and
$\Delta^{\Omega(1/\alpha)}$ colors for $O(n^{1+\alpha})$ space). They further
showed that allowing \emph{multiple} semi-streaming passes over the stream
makes better tradeoffs possible: one can get an $O(\Delta^2)$-coloring in $2$
passes, and an $O(\Delta)$-coloring in $O(\log \Delta)$ passes.  On the other
hand, Chakrabarti, Ghosh, and Stoeckl~\cite{ChakrabartiGS22}, considered a
``middle ground'' between deterministic and randomized algorithms, namely the
\emph{adversarially robust} algorithms introduced by~\cite{BenEliezerJWY20}. These algorithms work even when stream updates are generated by an adaptive adversary, depending on the
algorithm's previous outputs (and thus implicitly on its internal randomness;
observe that deterministic algorithms are always robust). They showed that a
(possibly randomized) robust semi-streaming coloring algorithm requires
$\Omega(\Delta^2)$ colors, while an $O(\Delta)$-coloring admits no
$o(n\Delta)$-space robust algorithm. The same work also gave a robust
semi-streaming algorithm achieving $O(\Delta^3)$ colors.  Thus, the results
in~\cite{AssadiCS22,ChakrabartiGS22} establish a neat trichotomy for
single-pass semi-streaming graph coloring: (i)~a $(\Delta+1)$-color palette
suffices for standard randomized streaming; (ii)~$\poly(\Delta)$ colors are
necessary and sufficient for robust streaming; and~(iii) $\exp(\Delta)$ colors
are needed for deterministic algorithms. 

Many questions in this line of work, however, remain unresolved. Here are two
particular ones: 
\begin{enumerate}[label=(\roman*),font=\rm] \em
  \item For deterministic algorithms, how many passes are needed to
  achieve a tight $(\Delta+1)$-coloring?
  \item For robust algorithms, where in the range
  $[\Delta^2,\Delta^3]$ does the above \textup{``$\poly(\Delta)$''} number of
  colors lie?
\end{enumerate}

This paper takes steps towards resolving both these questions.

\subsection{Our Contributions} 

\paragraph{The Deterministic Setting} In this setting, our main result is a
multi-pass algorithm for $(\Delta+1)$-coloring that runs in semi-streaming
space.
 
\begin{theorem} \label{thm:det-multipass-coloring}
  There is an efficient deterministic semi-streaming algorithm to
  $(\Delta+1)$-color an $n$-vertex graph, given a stream of its edges arriving in an adversarial order. The
  algorithm uses $O(n \log^2 n)$ bits of space and runs in $O(\log\Delta
  \log\log\Delta)$ passes. 
\end{theorem}

The above result shows that we can improve the $O(\Delta)$-coloring result
of~\cite{AssadiCS22} to the combinatorially optimal $(\Delta+1)$-coloring by
paying only an additional $O(\log\log{\Delta})$ factor in the number of passes. It
is worth pointing out here that in the streaming model, as well as several
other cases, it is known that $O(\Delta)$-coloring is an ``algorithmically
\emph{much} easier'' problem than $(\Delta+1)$-coloring. For instance, there are
quite simple single-pass randomized algorithms known for $O(\Delta)$-coloring \cite{BeraG18, AssadiCK19},
whereas the only known streaming $(\Delta+1)$-coloring algorithm, due
to~\cite{AssadiCK19}, uses sophisticated tools and a combinatorially involved
analysis.\footnote{Similar examples of this difference appear in the
(randomized) LOCAL algorithms~\cite{SchneiderW10,ChangLP18}, (deterministic)
dynamic graph algorithms~\cite{BhattacharyaCHN18}, or even provable
separations for the ``palette sparsification''
technique~\cite{AssadiCK19,AlonA20}. Yet another example is the closely
related problem of $O(\text{degeneracy})$-coloring versus
$(\text{degeneracy}+1)$-coloring studied by~\cite{BeraCG20} who proved that
the former admits a (randomized) single-pass semi-streaming algorithm while
the latter does not.}

Our algorithm in~\Cref{thm:det-multipass-coloring} uses a variety of novel
ideas and techniques. It is inspired by a recent distributed algorithm of
Ghaffari and Kuhn~\cite{GhaffariK21} that solves $(\Delta+1)$-coloring in the
CONGEST model. That algorithm was in turn inspired by earlier algorithms of
\cite{BambergerKM20} and \cite{Kuhn20}. We build on these works with non-trivial
modifications, additional methodology, and careful analyses. In particular, we
must contend with the limitation that the semi-streaming model does not allow
enough space for a typical vertex to ``know'' much of its neighborhood; this
is in sharp contrast to distributed computing models (including CONGEST).
Moreover, our algorithm achieves roughly $O(\log\Delta)$ passes, whereas the
\cite{GhaffariK21} algorithm uses $O(\log^2 \Delta\log n)$ distributed rounds;
this quantitative difference stems, in part, from our delicate tuning of
parameters in an iterative process that colors vertices in batches.

As a by-product of the technology developed for establishing
\Cref{thm:det-multipass-coloring}, we also obtain a similarly efficient
algorithm for the more general problem of $(\text{degree}+1)$-list-coloring.
In this problem, the input specifies a graph $G$ as usual and, for each vertex
$x$, a list $L_x$ of at least $\deg(x)+1$ allowed colors for $x$; the goal is
to properly color $G$ subject to these lists. In a streaming setting, the
input is a sequence of tokens, each either an edge of $G$ or a pair $(x,L_x)$
for some vertex $x$; these tokens may be interleaved arbitrarily. We obtain
the following algorithmic result.

\begin{theorem} \label{thm:d1lc-det-multipass}
  Let $C$ be a set of colors of size $O(n^2)$. There is a deterministic
  semi-streaming algorithm for $(\text{degree}+1)$-list-coloring a graph $G$
  given a stream consisting of, in any order, the edges of $G$ and $(x,L_x)$
  pairs specifying the list $L_x$ of allowed colors for a vertex $x$,
  where $L_x \subseteq C$. The algorithm uses $O(n \log^2 n)$ bits of space
  and runs in $O(\log\Delta \log\log\Delta)$ passes. 
\end{theorem}
 
\paragraph{The Adversarially Robust Setting} In this setting, our algorithm
needs to be correct against an adversary who constructs the input graph
adaptively by inserting upcoming edges based on the colorings returned by the
algorithm. This is inherently a single-pass setting. However, we are now
allowed to use randomness.  The interaction with the active adversary means
that the stream elements might depend on past outputs, which in turn depend on
the random bits used by the algorithm. While \cite{AssadiCK19} gave a
semi-streaming $(\Delta+1)$-coloring algorithm in the ``non-robust'' 
setting, where the stream is fixed in advance, \cite{ChakrabartiGS22}
showed that a \emph{robust} semi-streaming algorithm must use $\Omega(\Delta^2)$
colors. Our main result in the robust setting is the following.

\begin{theorem}\label{thm:adv-2.5-coloring}  
  There is an $O(\Delta^{5/2})$-coloring algorithm which is robust (with total
  error probability $\le \delta$) against adaptive adversaries, and runs in
  $O(n \log^{O(1)} n \cdot \log\delta^{-1})$ bits of space, given oracle access to
  $O(n \Delta)$ bits of randomness.
  %
  %
\end{theorem}

The above result improves a robust algorithm of \cite{ChakrabartiGS22}, which
runs in a similar semi-streaming amount of space but only gives an
$O(\Delta^3)$-coloring. Further, our robust algorithm admits a smooth tradeoff
between the memory used and the number of colors. Setting the parameters
appropriately, we can improve upon a different robust algorithm of
\cite{ChakrabartiGS22} that gives an $O(\Delta^2)$-coloring using
$O(n\Delta^{1/2})$ space. Restricted to $O(\Delta^2)$ colors, we can improve
the space usage to $O(n\Delta^{1/3})$.  On the other hand, given
$O(n\Delta^{1/2})$ space, we can reduce the number of colors to
$O(\Delta^{7/4})$.

Our algorithm overcomes the challenges posed by the adaptive adversary by
crucially exploiting the graph structure and cleverly using modified versions
of the known techniques on subgraphs of the input graph. These techniques
include those in the adversarially robust literature, such as sketch
switching~\cite{BenEliezerJWY20, ChakrabartiGS22}, as well as those in the coloring literature,
such graph partitioning and degeneracy-based coloring~\cite{BeraCG20}. 

One caveat of the above result is the need for a large number of random bits. The same caveat applies to the aforementioned robust $O(\Delta^3)$-coloring algorithm of \cite{ChakrabartiGS22}. One could argue that, in practice, this is surmountable by using a cryptographic pseudorandom generator. However, if we wish to stick to the mathematical definition of adversarial robustness (which is an information-theoretic security guarantee), we can still obtain an improvement over past work, as shown in the following result.

\begin{theorem} \label{thm:rand-efficient-robust}
  There is an adversarially robust $O(\Delta^3)$-coloring algorithm that runs in semi-streaming space, even including the random bits used by the algorithm.
\end{theorem}

\subsection{Related work}

The study of graph coloring in the classical streaming model was initiated parallelly
and independently by Bera and Ghosh~\cite{BeraG18} and Assadi, Chen, and Khanna~\cite{AssadiCK19}. The former work
obtained an $O(\Delta)$-coloring algorithm in semi-streaming
space, while the latter achieved a tight $(\Delta+1)$-coloring in the same
amount of space. The latter work uses an elegant framework called
\emph{palette sparsification}: each node samples a list of roughly $\log n$
colors from the palette of size $\Delta+1$, and it is shown that w.h.p.~there exists
a proper list-coloring where each node uses a color only from its list. This
immediately gives a semi-streaming $(\Delta+1)$-coloring algorithm since one
can store only ``conflicting'' edges that can be shown to be only $\tO(n)$
many w.h.p.\footnote{The algorithm that is immediately implied is an
exponential-time one where one can store the conflicting edges and obtain the
list-coloring by brute force. An elaborate method was then needed to implement
it in polynomial time.} This framework implying semi-streaming coloring
algorithms was then explored by Alon and Assadi~\cite{AlonA20} under various
palette sizes (based on multiple color parameters) as well as list sizes.
Their results also implied interesting algorithms for coloring triangle-free
graphs and for (degree+1)-list coloring. 

Very recently, Assadi, Chen, and Sun \cite{AssadiCS22} studied deterministic $\Delta$-based
coloring and showed that for a single pass, no non-trivial streaming algorithm
can be obtained. For semi-streaming space, any deterministic algorithm needs
$\exp(\Delta^{\Omega(1)}))$ colors, whereas for $O(n^{1+\alpha})$ space,
$\Delta^{\Omega(1/\alpha)}$ colors are needed.  Observe that these bounds are
essentially matched by the trivial algorithm that stores the graph when
$\Delta \leq n^{\alpha}$ in order to $(\Delta+1)$-color it at the end; or just
color the graph trivially with $n = \Delta^{1/\alpha}$ colors, without even
reading the edges, when $\Delta > n^{\alpha}$.  In light of this, a natural
approach is to consider the problem allowing multiple passes over the input
stream. They show that in just one additional pass, an $O(\Delta^2)$-coloring
can be obtained deterministically, while with $O(\log \Delta)$ passes, we can
have a deterministic $O(\Delta)$-coloring algorithm. Another very recent work
on $\Delta$-based coloring is that of Assadi, Kumar, and Mittal \cite{AssadiKM22}, who surprisingly
proved Brooks's theorem in the semi-streaming setting: any (connected) graph
that is not a clique or an odd cycle can be colored using exactly $\Delta$
colors in semi-streaming space. 

Other works on streaming coloring include the work of Abboud, Censor-Hillel, Khoury, and Paz~\cite{AbboudCKP19} who show that coloring an $n$-vertex graph with the optimal chromatic number of colors requires $\Omega(n^2/p)$ space in $p$ passes. They also show that deciding $c$-colorability for $3 \leq c < n$ (that might be a function of $n$) needs $\Omega((n - c)^2/p)$ space in $p$ passes. Another notable work is that of Bera, Chakrabarti, and Ghosh \cite{BeraCG20}, who
considered the problem with respect to the \emph{degeneracy} parameter that
often yields more efficient colorings, especially for sparse graphs. They
designed a semi-streaming $\kappa(1+o(1))$-coloring algorithm for graphs of
degeneracy~$\kappa$. They also proved that a combinatorially tight
$(\kappa+1)$-coloring is not algorithmically possible in sublinear space. In particular,
semi-streaming coloring needs $\kappa+\Omega(\sqrt{\kappa})$ colors. Bhattacharya, Bishnu, Mishra, and Upasana~\cite{BhattacharyaBMU21} showed that verifying whether an input vertex-coloring of a graph is proper is hard in the vertex-arrival streaming model where each vertex arrives with its color and incident edges. Hence, they consider a relaxed version of the problem that asks for a $(1\pm \epsilon)$-estimate of the number of conflicting edges. They prove tight bounds for this problem on adversarial-order streams and further study it on random-order streams. Recently, Halldorsson, Kuhn, Nolin, and Tonayan~\cite{HalldorssonKNT22} gave a palette-sparsification-based semi-streaming
algorithm for $(\textrm{degree}+1)$-list-coloring for any arbitrary list of colors
assigned to the nodes, improving upon the work of \cite{AlonA20} whose algorithm works only when the color-list of each vertex $v$ is
$\{1,\ldots,\deg(v)+1\}$. Note that all the works mentioned above are in the ``static'' streaming model and all their algorithms, except those in \cite{AssadiCS22}, are randomized and non-robust.

Starting with the work of Ben-Eliezer, Jayaram, Woodruff, and Yogev~\cite{BenEliezerJWY20}, the adversarially robust
streaming model has seen a flurry of research in the last couple of years
\cite{BenEliezerY20, HassidimKMMS20, KaplanMNS21, BravermanHMSSZ21,
WoodruffZ21, AttiasCSS21, beneliezerEO21, ChakrabartiGS22, Cohen0NSSS22, Sto23}.
Chakrabarti, Ghosh, and Stoeckl \cite{ChakrabartiGS22} were the first to study
graph coloring in this model. They showed a separation between standard and
robust streaming coloring algorithms by establishing lower bounds of
(i)~$\Omega(\Delta^2)$ colors for robust semi-streaming coloring, and
(ii)~$\Omega(n\Delta)$ space for robust $O(\Delta)$-coloring. In fact, they prove a
smooth colors/space tradeoff: a robust $K$-coloring algorithm requires
$\Omega(n\Delta^2/K)$ space. On the upper bound side, they design an $O(\Delta^3)$-coloring robust algorithm in semi-streaming space, with oracle access to $\tO(n\Delta)$ many random bits. They also obtain an $O(\Delta^2)$-coloring in $\tO(n\sqrt{\Delta})$ space (including random bits used).

%% file: prelim.tex

\section{Preliminaries}

\paragraph{Notation} Throughout the paper, ``$\log$'' denotes the base-$2$
logarithm; $[n]$ denotes the set $\{1,\ldots,n\}$; $\FF_p$ is the finite field
with $p$ elements; $\indic_{\mathtt{cond}}$ is the indicator function for
condition $\mathtt{cond}$, i.e., it takes the value $1$ when $\mathtt{cond}$
is true, and $0$ otherwise; and the notation $a \in_R A$ means that $a$ is
drawn uniformly at random from the finite set $A$.

A graph $G = (V,E)$ typically has $n = |V|$ vertices. We may identify $G$ with
its set of edges, and write $\{u,v\} \in G$ to mean that $\{u,v\}$ is an edge
in $G$. For $B \subseteq E$, $\deg_B(x)$ denotes the degree of $x$ in the
graph formed by the edges in $B$. For $X \subseteq V$, $G[X]$ denotes the
subgraph of $G$ induced by $X$.

\paragraph{Adversarially Robust Streaming} In the \emph{static} streaming
setting, an algorithm operates on a long sequence $\ang{e_1,e_2,\ldots}$ of
elements, reading them in order. It may make multiple passes over the stream.
We typically aim to design a streaming algorithm with parameters $\delta$ and
$S$ as low as possible so that, for all possible input streams, it uses $\le
S$ bits of space and errs with probability $\le \delta$. If the algorithm is
deterministic, then $\delta = 0$, and we seek to minimize space usage subject
to correctness on all inputs. 

In the \emph{adversarial} setting, we assume that the algorithm is one party
to a game between it and an \emph{adversary}; the adversary produces a
sequence $\ang{e_1,e_2,\ldots}$ of elements, and can ask the algorithm to
report an intermediate output $o_i$ after each new element $e_i$. Unlike the
static setting, the next element $e_{i+1}$ produced by the adversary may
depend (possibly randomly\footnote{However, there is always a deterministic
adversary at least as effective as any randomized one at making the algorithm
fail.}) on the transcript $\ang{e_1,o_1,\ldots,e_i,o_i}$ of the game. The
algorithm is said to err if at least one of its outputs is incorrect for the
problem at hand. In this setting, we typically aim to find streaming
algorithms minimizing $S,\delta$, where here we want the algorithm to
(a)~never exceed $S$ bits of space and (b)~err with probability $\le \delta$,
for all possible adversaries.

\paragraph{Colorings} A \emph{partial coloring} of a graph $G = (V,E)$ using a
palette $\cC$ (any nonempty finite set) is a tuple $(U,\chi)$ where $U
\subseteq V$ is the set of uncolored vertices and $\chi \colon V \to \cC \cup
\{\bot\}$ is a function such that $\chi(x) = \bot \Leftrightarrow x \in U$.
(we may also simply refer to $\chi$ as the partial coloring). The coloring is
said to be \emph{proper} if, for all $\{u,v\} \in E$ such that $u \notin U$
and $v \notin U$, we have $\chi(u) \ne \chi(v)$. A \emph{coloring} of $G$ is a
partial coloring where $U = \emptyset$.

Given a graph-theoretic parameter $\psi$, the $\psi$-coloring (algorithmic)
problem asks one to determine a proper coloring of an input graph $G$ using a
palette of size $|\cC| \le \psi$. This paper focuses first on
$(\Delta+1)$-coloring and later on $\poly(\Delta)$-coloring. We also consider
the {\em list coloring} problem, wherein each $x \in V$ has an associated list
(really a set) $L_x \subseteq \cC$ and we are to find a coloring satisfying
$\chi(x) \in L_x$ for all $x$. Specifically, we study the problem
$(\deg+1)$-list-coloring, in which $|L_x| = \deg(x)+1$ for each $x$.

\paragraph{Hash Functions} We will use the following standard properties of
families of hash functions.  A hash family $\cH$ of functions $A \rightarrow
B$ is \emph{$k$-independent} if, for all distinct $a_1,\ldots,a_k \in A$, and
arbitrary $b_1,\ldots,b_k \in B$,
\begin{align*}
  \Pr_{h \in_R \cH} \big[ h(a_1) = b_1 \land
    \cdots \land h(a_k) = b_k \big]
  = 1/|B|^k \,.
\end{align*}
The family is \emph{$2$-universal} if, for all distinct $a_1,a_2 \in A$,
\begin{align*}
  \Pr_{h \in_R \cH} \big[ h(a_1) = h(a_2) \big] \le 1/|B| \,.
\end{align*}

\paragraph{Useful Lemmas} These variations of standard lemmas 
are proved in \Cref{sec:deferred-proofs}, for completeness.

\begin{lemma}[A constructive variation on Tur\'an's theorem]\label{lem:find-iset}
  Given a graph with $n$ vertices and $m$ edges, one can find an independent
  set of size $\ge n^2 / (2m + n)$ in deterministic polynomial time. 
\end{lemma}

\begin{lemma}[Mix of Chernoff bound and Azuma's inequality] \label{lem:forward-concentration}  
  Let $X_1,\ldots,X_k$ be a sequence of $\{0,1\}$ random variables, and $c \in
  [0,1]$ a real number for which, for all $i \in k$, $\EE[X_i \mid
  X_1,\ldots,X_{i-1}] \le c$. Then
  \begin{align*}
    \Pr\left[\sum_{i \in [k]} X_i \ge (1 + t) k c \right] \le 2^{- t k c} \,,
    \qquad \text{assuming $t \ge 3$.}
  \end{align*}
\end{lemma}

%% file: determ.tex

\section{A (Multipass) Deterministic Algorithm}

This section presents our first main result, giving a multipass deterministic
semi-streaming algorithm for $(\Delta+1)$-coloring, proving
\Cref{thm:det-multipass-coloring}. As usual, let $G = (V,E)$ denote the input
graph, which has $n = |V|$ vertices and maximum degree $\Delta$. Later, we
shall extend our algorithm to the $(\deg+1)$-list-coloring problem, so it will
be helpful to think of each vertex $x \in V$ being associated with a set $L_x$
of allowed colors; for the algorithm we discuss first, $L_x = [\Delta+1]$ for
each $x \in V$.

\subsection{High-Level Organization}

The algorithm's passes are organized as follows. The algorithm proceeds in
{\em epochs}, where each epoch starts with a partial coloring $\chi$ that has
a certain subset $U \subseteq V$ uncolored and ends with a new partial
coloring that extends $\chi$ by coloring at least a constant fraction of the
vertices in $U$, thereby shrinking $|U|$ to $\alpha |U|$, for some constant
$\alpha < 1$. In the beginning, $U = V$. After at most
$\ceil{\log_{1/\alpha}\Delta}$ such epochs, we will have $|U| \le n/\Delta$:
at this point, the algorithm makes a final pass to collect all edges incident
to $U$ and greedily extend $\chi$ to a full coloring of $G$.

Each epoch of the algorithm is divided into {\em stages}, where each stage
whittles down a set of proposed colors for each uncolored vertex. To explain
this better, the following definition is useful.

\begin{definition}[partial commitment, slack, potential] \label{def:slack}
  A {\em partially committed coloring} (PCC) of $G$ is an assignment of colors
  and lists to the vertices satisfying the following conditions.
  \begin{itemize}[topsep=2pt, itemsep=0pt]
    \item Every vertex outside a subset $U \subseteq V$ of uncolored vertices
    is assigned a specific color $\chi(x) \in L_x$; the resulting $\chi$ is a
    proper partial coloring.
    \item Each $x \in U$ has an associated set $P_x$ of
    proposed colors, defining a collection $\cP = \{P_x\}_{x \in U}$. 
    \item For every two vertices $x, y \in U$, either $P_x = P_y$ or $P_x \cap
    P_y = \emptyset$.
  \end{itemize}
  We shall denote such a PCC by the tuple $(U,\chi,\cP)$.
  Given such a PCC, define the {\em slack} of a vertex with respect to a set
  $T$ of colors by
  \begin{align} 
    \slackwrt{T}{x} &= \max\{0,\, |T \cap L_x| - 
      |\{y \in \Nhd(x) \setm U:\, \chi(y) \in T\}|\} \,,
    \label{eq:slack-def}
  \end{align}
  and further define $s_x = \slackwrt{P_x}{x}$; that is, $s_x$ is the number
  of colors in $P_x$ that are available to $x$ in $L_x$ minus the number of
  \emph{times} the colors in $P_x$ have appeared in the already colored
  neighbors of $x$.
  Define the {\em potential} of the PCC to be
  \begin{align}
    \Phi = \Phi(U,\chi,\cP)
    &= \sum_{\{x,y\} \in E} \indic_{x \in U \land y \in U} \cdot \indic_{P_x = P_y} \cdot 
      \left( \frac{1}{s_x} + \frac{1}{s_y} \right) \,
    \label{eq:potential-def}
  \end{align}
  which sums the quantity $(1/s_x + 1/s_y)$ over all edges $\{x,y\}$ inside
  $U$ with $P_x = P_y$. \qed
\end{definition}

Intuitively, the slack defined here is a lower bound on the number of unused
colors available to a vertex. Our definition differs slightly from the "slack" defined
by \cite{HalldorssonKNT22}, where the number of colors used by the
neighbors is known exactly.
It turns such a lower bound on the number of unused colors is sufficient for our algorithm to progressively refine a PCC. 
The advantage of this lower bound -- equivalently, of using an upper bound on the number
of used colors, $|\{y \in \Nhd(x) \setm U:\, \chi(y) \in T\}|$, instead of the
exact quantity $|T \cap \{\chi(y) :\, y \in \Nhd(x) \setm U \}|$ -- is that the former is
a linear function of the data stream, and can be easily computed in $O(\log n)$ space.
Meanwhile, as a consequence of the set disjointness lower bound in communication complexity,
determining the latter can require up to $\Omega(\Delta)$ space. In the LOCAL and CONGEST 
models, each vertex can easily store and maintain a list of all its available colors 
(equivalently, colors used by its neighborhood), so the algorithms of 
\cite{GhaffariK21,BambergerKM20} do not need such a modified notion of "slack".

The set $\Free(T,x) := T \cap L_x \setm \{\chi(y):\, y \in \Nhd(x) \setm U\}$
is the set of all colors in $T$ that are available for $x$, in light of the
local constraints imposed by $L_x$ and $\chi$. Notice that $|\Free(T,x)| \ge
\slackwrt{T}{x}$, since a color in $T$ might be used more than once in the
neighborhood of $x$, thus reducing the LHS only once, but the RHS more than
once. Hence, if we extend $\chi$ to a full coloring by choosing, independently
for each $x \in U$, a uniformly random color in $\Free(P_x,x)$, the only
monochromatic edges we might create are within $U$ and the number,
$m_{\mono}(U,\chi,\cP)$, of such edges satisfies
\begin{align}
  \EE m_{\mono}(U,\chi,\cP) 
  = \sum_{\ontop{\{x,y\} \in E(G[U])}{P_x = P_y}} 
    \frac{|\Free(P_x,x) \cap \Free(P_y,y)|}{|\Free(P_x,x)| \cdot |\Free(P_y,y)|}
  \le \sum_{\ontop{\{x,y\} \in E(G[U])}{P_x = P_y}} 
    \left( \frac{1}{s_x} + \frac{1}{s_y} \right)
  = \Phi \,. \label{eq:potential-bound}
\end{align}

\subsection{The Logic of an Epoch: Extending a Partial Coloring}

Returning to the algorithm outline, at the start of an epoch, the current
partial coloring $\chi$ and its corresponding set $U$ of uncolored vertices
define a trivial PCC where $P_x = L_x = [\Delta+1]$ for each $x$. We shall
eventually show that the resulting potential $\Phi \le |U|$. Each stage in the
epoch shrinks these sets $P_x$ in such a way that the potential $\Phi$ does
not increase much. After several stages, each $P_x$ in the PCC becomes a
singleton and the bound on $\Phi$, together with \cref{eq:potential-bound},
ensures that assigning each $x \in U$ the sole surviving color in $P_x$ would
not create too many monochromatic edges.  Now, \Cref{lem:find-iset} allows us
to commit to these proposed colors for at least $(1-\alpha)|U|$ of the
uncolored vertices; this defines a new partial coloring and ends the epoch.

We now describe how to shrink the sets $P_x$. For this, view each color as a
$b$-bit vector where $b = \ceil{\log(\Delta+1)}$ according to some canonical
mapping, e.g., $\ba \in \b^b \mapsto 1 + \sum_{i=1}^b a_i 2^{i-1}$. Each set
$P_x$ will correspond to a subcube of $\b^b$ where the first several bits have
been fixed to particular values.%
\footnote{If $\Delta+1$ is not a power of $2$, $P_x$ might contain elements
not in $L_x$, but this doesn't matter because $\Free(T,x) \subseteq L_x$
always.} %
Each stage of the $r$th epoch (except perhaps the last, due to divisibility
issues) will shrink each $P_x$ by fixing an additional $k$ bits of its
subcube, thus reducing the dimension of the subcube. We choose $k := 1 +
\floor{\log(n/|U|)}$, so that $|U| 2^{k} \le 2n$; this bound will be important
when we analyze the space complexity. The epoch ends when all bits of each
$P_x$ have been fixed, making each $P_x$ a singleton; clearly, this happens
after $\ceil{b/k}$ stages.

This brings us to the heart of the algorithm: we need to describe, for each $x
\in U$ and the particular value of $k$ for the current epoch, how to fix the
next $k$ bits for $P_x$. Let $P_{x,\bj}$ be the subset of $P_x$ where the $k$
lowest-indexed free bits are set to $\bj \in \b^k$: this partitions $P_x$ into
$2^k$ subcubes. Define
\begin{align}
  w_{x,\bj} = \frac{\slackwrt{P_{x,\bj}}{x}}
    {\sum_{\bi \in \b^k} \slackwrt{P_{x,\bi}}{x}} \,.
    \label{eq:wxj-def}
\end{align}
An easy calculation shows that if, for each $x$, we choose $\bj$ at {\em
random} according to the distribution given by $(w_{x,\bj})_{\bj
\in \b^k}$ to obtain a new random collection $\widetilde{\cP}$ of proposed
color sets for each vertex, then
\begin{align}
  \EE \Phi(U,\chi,\widetilde{\cP}) = \Phi(U,\chi,\cP) \,.
  \label{eq:pcc-refinement}
\end{align}
Therefore, there {\em exists} a particular realization
$\cP'$ of $\widetilde{\cP}$ such that $\Phi(U,\chi,\cP') \le
\Phi(U,\chi,\cP)$. However, it is not clear how to identify such a $\cP'$
deterministically and in a space-efficient manner in a stream. 

A key idea that enables a space-efficient derandomization is to choose the
$\bj$ values for the vertices $x \in U$ in a pseudorandom fashion, using a
$2$-independent family $\cH$ of hash functions $V \mapsto [p]$ for a
not-too-large value $p$. By using a suitable map $g \colon U \times [p] \to
\b^k$, we can use a uniform random value in $[p]$ to sample from a
distribution {\em close enough} to the $(w_{x,\bj})$ distribution. Then, for
each $x$, we shrink $P_x$ to $P_{x,\bj(x)}$ where $\bj(x) = g(x,h(x))$ and $h
\in_R \cH$. Let $\cP_h$ denote the resulting collection of proposed color
sets.

It turns out that a prime $p = \Theta(n \log n)$ suffices for the guarantees
we will eventually need. Thus, by choosing (e.g.) the Carter--Wegman family of
affine functions on $\FF_p$, we can take $|\cH| = O(n^2 \log^2 n)$. This
enables us to use two streaming passes with $\tO(n)$ space to
identify a specific function $h \in \cH$ that is ``approximately best'' in the
sense of minimizing $\Phi(U,\chi,\cP_h)$. We will then show that the new
potential is at most $1 + O(1/\log n)$ times the old. Repeating this argument
for each of the $O(\log n)$ stages in the epoch shows that at the end of the
epoch, the potential will have increased by at most a constant factor which
will then allow us to shrink $U$ by a constant factor $\alpha$, as noted
earlier.

The above outline suggests $O(\log n)$ epochs, each using $O(\log n)$ stages,
each of which uses $O(1)$ passes. Later, we shall show that a more careful
analysis bounds the number of passes by $O(\log\Delta \log\log\Delta)$.

\subsection{Detailed Algorithm and Proof of Correctness}

We now describe the algorithm more formally, by fleshing out the precise logic
of an epoch. Let $\cQ^{(i)}$ denote the partition of the color space $\b^b$
into subcubes $\smash{Q^{(i)}_{\bj}}$ defined by setting the $i$th $k$-bit
block to each of the $2^k$ possible patterns $\bj$; i.e.,
\begin{align}
  Q^{(i)}_{\bj} := \big\{\ba \in \b^b: (a_{ki-k+1}, \ldots, a_{ki}) = \bj\big\} \,; \quad
  \cQ^{(i)} := \big\{ Q^{(i)}_{\bj} \big\}_{\bj \in \b^k} \,.
  \label{eq:qij-def}
\end{align}
If $k$ does not divide $b$, we must make an exception for the $\ceil{b/k}$th
partition, for which the relevant bit patterns $\bj$ would be shorter; for
clarity of presentation, we shall ignore this edge case in what follows.

Before we proceed, we also need the following lemma, whose proof is given in \Cref{sec:deferred-proofs}.
\begin{lemma}\label{lem:gw-property}
For $p \ge 8 n \log n$, and $\bw = (w_{x,\bj})_{x \in U, \bj \in \b^k}$ there is a function $g_{\bw} \colon U \times [p] \to \b^k$ satisfying:
    \begin{align*}
          \frac{|g_{\bw}^{-1}(x,\bj)|}{p} \le 
          w_{x,\bj} \left(1 + \frac{1}{8\log n}\right) \,, \quad
          \forall~ \bj \in \b^k
    \end{align*}
\end{lemma}

The full logic of the algorithm is given in \Cref{alg:determ}.

\begin{algorithm*}[!htbp]
\begin{algorithmic}[1]
\Procedure{Deterministic-Coloring}{streamed $n$-vertex graph $G = (V,E)$ with max degree $\Delta$}
    \State $U \gets V$;~ $\chi(x) \gets \bot$ for all $x \in V$ \Comment{all vertices uncolored}
    \Repeat
        \State \Call{Coloring-Epoch}{$G, U, \chi$} \Comment{shrinks $|U|$ to at most $\alpha|U|$}
    \Until{$|U| \le n/\Delta$}
    \State In one pass, collect every edge incident to a vertex in $U$ \label{line:final-pass} 
    \State Use these edges to greedily complete $\chi$ to a proper coloring of $G$
\EndProcedure
\Statex

\Procedure{Coloring-Epoch}{graph $G$, partial coloring $(U,\chi)$}

    \State $b \gets \ceil{\log(\Delta+1)}$ \Comment{each color is a $b$-bit vector}

    \State $k \gets 1 + \floor{\log\left(n/|U|\right)}$ \Comment{number of bits fixed in each stage}

    \ForAll{$x \in U$} $P_x \gets \b^b$ 
        \Comment{the initial, trivial PCC}
    \EndFor

    \ForAll{stage $i$, from $1$ through $\ceil{b/k}$}
        \State \textbf{pass 1:}
        \Indent
            \ForAll{$x \in U$ and $Q \in \cQ^{(i)}$}
                compute $\slackwrt{P_x \cap Q}{x}$ by
                using~\cref{eq:slack-def}
                \EndFor
        \EndIndent

        \State Determine all \smash{$w_{x,\bj}$} values using \cref{eq:wxj-def}, 
        noting that \smash{$P_{x,\bj} = P_x \cap Q^{(i)}_\bj$}

        \State $p \gets$ prime in $[8n\log n, 16n\log n]$;~
        $\cH \gets \{z \mapsto az+b:\, a,b\in\FF_p\}$ \label{line:cw-hash}
            \Comment{Carter--Wegman hashing}

        \State Implicitly construct $g_{\bw} \colon U \times [p] \to \b^k$ as per \Cref{lem:gw-property}.
        
        \State For each $h \in \cH$, define $\cP_h = \{P_{x,h}\}_{x \in U}$, where $P_{x,h} := P_{x} \cap Q^{(i)}_{g_{\bw}(x,h(x))}$ \label{line:new-potential} 

        \LineComment{Identify a specific $h^\star \in \cH$ for which
        $\Phi(U,\chi,\cP_{h^\star})$ is not much larger than average, as follows:}

        \State \textbf{pass 2:}
        \Indent
            \State Split $\cH$ into $\sqrt{|\cH|}$ parts 
            \State Estimate $\sum_{h} \Phi(U,\chi,\cP_h)$ for each part, 
            up to $(1+1/(8\log n))$ relative error \label{line:est-best-hash-part}
            \State Pick the part minimizing the estimated sum
        \EndIndent

        \State \textbf{pass 3:}
        \Indent
            \State Estimate $\Phi(U,\chi, h)$ for each $h$ within the chosen part, 
              up to $(1+1/(8\log n))$ relative error \label{line:est-best-hash-func}
            \State Choose $h^\star$ as the (approximate) minimizer
        \EndIndent
        
        \ForAll{$x \in U$} $P_x \gets P_{x, h^\star}$ \label{line:end-of-stage}
            \Comment{constrain the PCC more tightly}
        \EndFor
    \EndFor
    \Statex
    
    \State \textbf{end-of-epoch pass:} \Comment{each $P_x$ is now a singleton}
    \Indent
        \State Collect $F \gets \{\{u,v\} \in E:\, u \in U, v \in U,$ and $P_u = P_v\}$
            \Comment{we will prove that $|F| = O(|U|)$}
        \State In the graph $(V,F)$, find an independent set $I$ with $|I| \ge (1-\alpha)|U|$, using \Cref{lem:find-iset}
        \ForAll{$x \in I$} \Comment{extend $\chi$ by coloring $x$}
            \State $U \gets U \setm \{x\}$
            \State $\chi(x) \gets$ the sole element in $P_x$
        \EndFor
    \EndIndent
\EndProcedure
\end{algorithmic}
\caption{Deterministic Semi-Streaming Algorithm for $(\Delta+1)$-Coloring}
\label{alg:determ}
\end{algorithm*}

The most important aspect of the analysis is to quantify the progress made in
each epoch and establish that the colors proposed at the end of each stage do
not produce too many monochromatic edges (i.e., those in $F$.) This analysis
will demonstrate the utility of the potential defined in
\cref{eq:potential-def}.

Given a PCC $(U,\chi,\cP)$ where $\cP = \{P_x\}_{x\in U}$, define the
``conflict degree'' of each $x \in U$ by
\begin{align}
  \dconf(x) = \dconf(x; U,\chi,\cP) 
  := |\{y \in \Nhd(x) \cap U :\, P_y = P_x \}| \,,
  \label{eq:dconf-def}
\end{align}
which counts the neighbors of $x$ that could potentially form monochromatic
edges with $x$, were we to assign colors from $\cP$ to the uncolored vertices.
Recall the quantities $s_x = \slackwrt{P_x}{x}$ from \Cref{def:slack}.

\begin{lemma} \label{lem:dconf}
  For every PCC, $\Phi(U,\chi,\cP) = \sum_{x \in U} \dconf(x)/s_x$.
\end{lemma}  
\begin{prf}
  From the definitions in \cref{eq:slack-def,eq:potential-def}, using some 
  straightforward algebra,
  \[
    \Phi(U,\chi,\cP)
    = \sum_{\ontop{\{u,v\} \in E(G[U])}{P_u = P_v}} 
      \left( \frac{1}{s_u} + \frac{1}{s_v} \right) 
    = \sum_{x \in U} \frac{|\{y \in U:\, \{x,y\} \in E \land P_x = P_y\}|}{ s_x }
    = \sum_{x \in U} \frac{\dconf(x)}{s_x} \,. \qedhere
  \]
\end{prf}

\begin{lemma} \label{lem:slack-subadditive}
  For all $x$ and disjoint sets $T_1, T_2$:
  $\slackwrt{T_1 \sqcup T_2}{x} \le \slackwrt{T_1}{x} + \slackwrt{T_2}{x}$.
\end{lemma}
\begin{prf}
  This is straightforward from \cref{eq:slack-def} and the fact that
  $\max\{0,a_1+a_2\} \le \max\{0,a_1\} + \max\{0,a_2\}$.
\end{prf}

\begin{lemma} \label{lem:potential-growth}
  Suppose we start a particular epoch with the partial coloring $(U,\chi)$ and
  the initial, trivial PCC $(U,\chi,\cP_0)$. Suppose there are $\ell$ stages
  in this epoch and the $i$th stage begins with the PCC $\cP_i$. Let $\Phi_i
  := \Phi(U,\chi,\cP_i)$ be the corresponding potential, for $0 \le i \le
  \ell$.  Then $\Phi_0 \le |U|$ and $\Phi_\ell \le 2|U|$.
\end{lemma}
\begin{proof}
  Recalling that each $L_x \cap P_x = L_x = [\Delta+1]$ for the initial PCC,
  we use \cref{eq:slack-def,eq:dconf-def} to derive
  \[
    s_x - \dconf(x) 
    = \max\{0, \Delta + 1 - |\Nhd(x) \setm U|\} - |\Nhd(x) \cap U|
    = \Delta + 1 - \deg (x)
    \ge 1 \,.
  \]
  Thus, $\dconf(x) / s_x \le 1$ (and is not ``$0/0$'') for all $x \in U$.
  \Cref{lem:dconf} now implies $\Phi_0 \le |U|$.

  We now argue that, between each pair of successive stages, the potential
  $\Phi_i$ does not increase by much. First observe that when $h$ is drawn
  uniformly at random from $\cH$, and $u\ne v$,
  \begin{align}
    \Pr\left[ g_{\bw}(u,h(u)) = g_{\bw}(v,h(v)) = \bj \right] 
    &= \Pr\left[ g_{\bw}(u,h(u)) = \bj \right] \cdot \Pr\left[ g_{\bw}(v,h(v)) = \bj \right] \notag\\
    &= \Pr\left[ h(u) \in g_{\bw}^{-1}(u,\bj) \right] \cdot \Pr\left[h(v) \in g_{\bw}^{-1}(v,\bj) \right] \notag\\
    & \le w_{u,\bj} w_{v,\bj} \left(1 + \frac{1}{8 \log n}\right)^2 \notag\\
    &\le e^{1/(4 \log n)} w_{u,\bj} w_{v,\bj} \,. \label{eq:collision-prob}
  \end{align}

  \renewcommand{\slackwrt}[2]{\textup{sk}(#2 \mid #1)} To keep the rest the
  derivation compact, let us abbreviate ``slack'' to ``sk.'' 
   The candidate PCCs $\cP_h$ defined in
  \cref{line:new-potential} are tightenings of the current PCC in which we pick subcubes according to the specific hash function $h$. With $h$ chosen uniformly
  at random from $\cH$:
  \newcommand*{\Widest}{\expln{lemma~3.3~~}{=}}%
  \newcommand*{\Center}[1]{\makebox[\widthof{$\Widest$}]{$#1$}}%
  \begin{alignat}{1}
    \EE \Phi(U,\chi,\cP_h)
    &\Center{=}
      \sum_{\{u,v\} \in E} \EE \indic_{u \in U} \indic_{v \in U} \cdot \indic_{P_{u,h} = P_{v,h}} \cdot 
      \left(\frac{1}{ \slackwrt{P_{u,h}}{u}} + \frac{1}{\slackwrt{P_{v,h}}{v}} \right) \notag\\
    &\Center{=}
      \sum_{\{u,v\} \in E(G[U])} \sum_{\bj \in \b^k} 
      \Pr\left[ P_{u,h} = P_{u,\bj} = P_{v,\bj} = P_{v,h} \right] 
      \left(\frac{1}{ \slackwrt{P_{u,\bj}}{u}} + \frac{1}{\slackwrt{P_{v,\bj}}{v}} \right) \notag\\ 
    &\Center{\expln{\cref{line:new-potential}}{=}}
      \sum_{\{u,v\} \in E(G[U])} \sum_{\bj \in \b^k} \indic_{P_u = P_v}
      \Pr\left[ g_{\bw}(u,h(u)) = g_{\bw}(v,h(v)) = \bj \right] 
      \left(\frac{1}{ \slackwrt{P_{u,\bj}}{u}} + \frac{1}{\slackwrt{P_{v,\bj}}{v}} \right) \notag\\
    &\Center{\expln{\cref{eq:collision-prob}}{\le}}
      \sum_{\ontop{\{u,v\} \in E(G[U])}{P_u = P_v}} 
      \sum_{\bj \in \b^k} e^{1/(4 \log n)} w_{u,\bj} w_{v,\bj} 
      \left(\frac{1}{ \slackwrt{P_{u,\bj}}{u}} + \frac{1}{\slackwrt{P_{v,\bj}}{v}} \right) \notag\\
    &\Center{\expln{\cref{eq:wxj-def}}{=}}
      e^{1/(4 \log n)}  \sum_{\ontop{\{u,v\} \in E(G[U])}{P_u = P_v}} \sum_{\bj \in \b^k} 
      \frac{\slackwrt{P_{u,\bj}}{u}}{\sum_{\bi}\slackwrt{P_{u,\bi}}{u}} \cdot
      \frac{\slackwrt{P_{v,\bj}}{v}}{\sum_{\bi}\slackwrt{P_{v,\bi}}{v}} \cdot
      \left(\frac{1}{ \slackwrt{P_{u,\bj}}{u}} + \frac{1}{\slackwrt{P_{v,\bj}}{v}} \right) \notag\\
    &\Center{=}
      e^{1/(4 \log n)}  \sum_{\ontop{\{u,v\} \in E(G[U])}{P_u = P_v}}
      \sum_{\bj \in \b^k} \frac{\slackwrt{P_{u,\bj}}{u} + \slackwrt{P_{v,\bj}}{v}}
      {\sum_{\bi}\slackwrt{P_{u,\bi}}{u} \cdot \sum_{\bi}\slackwrt{P_{v,\bi}}{v}} \notag\\
    &\Center{=}
      e^{1/(4 \log n)}  \sum_{\ontop{\{u,v\} \in E(G[U])}{P_u = P_v}}
      \left(\frac{1}{\sum_{\bj}\slackwrt{P_{u,\bj}}{u}} +
      \frac{1}{\sum_{\bj}\slackwrt{P_{v,\bj}}{v}}\right) \notag\\
    &\Center{\expln{\cref{lem:slack-subadditive}}{\le}}
      e^{1/(4 \log n)}  \sum_{\ontop{\{u,v\} \in E(G[U])}{P_u = P_v}}
      \left(\frac{1}{ \slackwrt{P_{u}}{u}} +
      \frac{1}{\slackwrt{P_{u}}{v}}\right) \notag\\
    &\Center{=}
      e^{1/(4 \log n)} \Phi_i \,. \label{eq:potential-growth-stage}
  \end{alignat}

  Thus, picking $h^\star$ with $\Phi(U,\chi,\cP_{h^\star})$ below average would ensure
  $\Phi_{i+1} \le e^{1/(4 \log n)} \Phi_{i} $. However, due to precision
  constraints, each of \cref{line:est-best-hash-part,line:est-best-hash-func}
  could contribute a relative error of $(1+1/(8\log n))$, so the $h^\star$
  actually picked by the algorithm gives only the following weaker guarantee:
  \begin{align*}
    \Phi_{i+1} 
    \le \left(1 + \frac{1}{8 \log n}\right)^2 e^{1/(4 \log n)} \Phi_{i}
    \le e^{1/(2 \log n)} \Phi_{i} \,.
  \end{align*}

  Since the number of stages in the epoch is $\ell \le \ceil{b/k} \le
  \log(\Delta + 1) \le \log n$, we have
  \[
    \Phi_\ell 
    \le \left(e^{1/(2 \log n)}\right)^{\ell} \Phi_0
    \le e^{1/2} |U| \le 2|U|\,. \qedhere
  \]
\end{proof}

The crucial combinatorial property of the $(\Delta+1)$-coloring problem is
that given any proper partial coloring, every uncolored vertex is guaranteed
to have a free color not in use by its colored neighbors. The next lemma
argues that even as we gradually tighten constraints in our PCC during the
stages of an epoch, a similar guarantee is maintained.

\begin{lemma} \label{lem:slack-positive}
  In each epoch, for all $x \in U$, the stages maintain the invariant that
  $s_x \ge 1$ and after the last stage we have $s_x = 1$.
\end{lemma}
\begin{proof}
  At the start of the epoch, $s_x \ge |L_x| - |\Nhd(x)| = (\Delta+1) - \deg (x)
  \ge 1$. 

  Consider a particular stage, which begins with a PCC $(U,\chi,\cP)$, where
  $\cP = \{P_{x}\}_{x\in U}$. Fix a vertex $x \in U$. In the next PCC formed
  at the end of the stage, $P_x$ shrinks down to $P_{x,\bj} = P_x \cap
  Q^{(i)}_{\bj}$ for a pattern $\bj \in \b^k$ satisfying $w_{x,\bj} > 0$: the
  way $g_{\bw}$ is defined (\Cref{lem:gw-property}) ensures this. By
  \Cref{lem:slack-subadditive},
  \[
    \sum_{\bi \in \b^k} \slackwrt{P_{x,\bi}}{x} 
    \ge \slackwrt{P_{x}}{x} = s_x \ge 1 \,,
  \]
  so there exists $\bj \in \b^k$ for which $\slackwrt{P_{x,\bj}}{x} \ge 1$.
  One such $\bj$ must be picked as the chosen pattern for $x$, because
  $w_{x,\bj} > 0$ implies $\slackwrt{P_{x,\bj}}{x} > 0$.  Consequently, the
  new value of $P_{x}$ chosen at the end of the stage
  (\cref{line:end-of-stage}) will continue to satisfy the invariant $s_x \ge
  1$.

  After the last stage in the epoch, every set $P_x$ is a singleton because,
  in the corresponding subcube of $\b^b$, all bits have been fixed. It is not
  possible that $P_x$ is empty, because $|P_x \cap L_x| \ge s_x \ge 1$. Thus
  $|P_x \cap L_x| = s_x = 1$.
\end{proof}

\begin{lemma} \label{lem:end-of-epoch}
  The set $F$ collected at the end of an epoch satisfies $|F| \le |U|$.
\end{lemma}
\begin{proof}
  Using the terminology of \Cref{lem:potential-growth}, at the end of an
  epoch, we have
  \newcommand*{\Widest}{\expln{~~lemma~3.3~~}{=}}%
  \newcommand*{\Center}[1]{\makebox[\widthof{$\Widest$}]{$#1$}}%
  \[
    2|U|
    \Center{\expln{\cref{lem:potential-growth}}{\ge}}
      \Phi_\ell 
    \Center{\expln{\cref{lem:dconf}}{=}}
      \sum_{x \in U} \frac{\dconf(x)}{s_x} 
    \Center{\expln{\cref{lem:slack-positive}}{=}}
      \sum_{x \in U} \frac{|\{y \in \Nhd(x) \cap U:\, P_x = P_y\}|}{1} 
    = 2 |F| \,. \qedhere
  \]
\end{proof}

\begin{lemma} \label{lem:epoch-progress}
  Each epoch maintains the invariant that $(U,\chi)$ is a proper partial
  coloring and shrinks the set of uncolored vertices $U$ to a smaller $U'$
  with $|U'| \le \frac23 |U|$.
\end{lemma}
\begin{proof}
  As noted before, at the end of the epoch, each set $P_x$ is a singleton and
  the sole color $c_x \in P_x$ is not used in $\Nhd(x)$ because $s_x \ne 0$
  (\Cref{lem:slack-positive}). Therefore, the set $F$ collected at the end is
  precisely the set of edges that would be monochromatic \emph{if} we colored
  each $x \in U$ with $c_x$. It follows that the end-of-epoch logic in the
  algorithm, which commits to these colors only on an \emph{independent} set
  in the graph $(V,F)$, maintains the invariant of a proper partial coloring.

  By \Cref{lem:find-iset}, $(V,F)$ contains an independent set $I$ of size
  \newcommand*{\Widest}{\expln{~~lemma~3.3~~}{=}}%
  \newcommand*{\Center}[1]{\makebox[\widthof{$\Widest$}]{$#1$}}%
  \[
    |I| \ge \frac{|U|^2}{2|F| + |U|} 
    \Center{\expln{\cref{lem:end-of-epoch}}{\ge}} \frac{|U|}{3} 
  \]
  and one can compute $I$ in polynomial time. Therefore, $|U'| = |U| - |I| \le
  \frac23 |U|$.
\end{proof}

\subsection{Space and Pass Complexity}

\begin{lemma} \label{lem:determ-space-bound}
  \Cref{alg:determ} runs in $O(n \log^2 n)$ bits of space and $O(\log\Delta \cdot 
  \log\log\Delta)$ streaming passes.
\end{lemma}
\begin{proof}
  For the space bound, it suffices to establish that \Call{Coloring-Epoch}{}
  runs in $O(n \log^2 n)$ space. At each stage of an epoch, the algorithm
  maintains the current PCC, consisting of the partial coloring $(U,\chi)$ and
  the collection $\cP = \{P_x\}_{x\in U}$. The former can be stored in
  $O(n\log\Delta)$ bits directly; so can the latter, since the subcube
  structure of $P_x$ allows for a natural $O(b) = O(\log\Delta)$-bit
  description.

  We now turn to the space required to execute the passes. Focus on stage~$i$
  within epoch~$r$. Computing the slack values in pass~$1$ requires $|U|2^k$
  counters, one for each pair $(x,Q^{(i)}_\bj)$, to determine $|\{y \in \Nhd(x):\, \chi(y) \in P_x \cap
  Q^{(i)}_\bj\}|$. Each such counter fits in $O(\log\Delta)$ bits. By our
  choice of $k$, the total space bound for these counters is $O(n\log\Delta)$.
  Moving on, identifying $h^\star$ requires keeping track of $\sqrt{|\cH|}$
  accumulators, to evaluate sums of the form given in
  \cref{line:new-potential}, in each of passes~$2$ and~$3$.  These
  accumulators do not need to be stored at full precision; a relative error of
  $(1+1/(8\log n))$ is acceptable, so $O(\log n)$ bits per accumulator
  suffice. Since $p = \Theta(\log n)$ and $|\cH| = p^2$ (\cref{line:cw-hash}), the total space cost of all the
  accumulators is $O(\sqrt{|H|} \log n) = O(n \log^2 n)$ bits.

  Next, we consider the end-of-epoch pass. 
  By~\Cref{lem:end-of-epoch}, $|{F}| \leq |{U}| = O(n)$ so this pass needs only $O(n\log
  n)$ bits to collect the edges in $F$. The rest of its computations happen
  offline and need no further storage. This completes the space complexity
  analysis.

  Finally, we account for the number of passes. In epoch~$r$, there are
  $\ceil{b/k_r}$ stages, where $k_r$ is the value of $k$ for the epoch; each
  such stage makes three streaming passes; additionally, there is one
  end-of-epoch pass. There is also one final pass after all epochs are done
  (\cref{line:final-pass}). By \Cref{lem:epoch-progress}, each epoch shrinks $|U|$ to at most $\alpha=2/3$ times its previous value. Notice that the epochs stop once $|U| \le
  n/\Delta$, so there are at most $\ceil{\log_{1/\alpha} \Delta}$ epochs.
  Furthermore, at the start of the $r$th epoch, $|U| \le \alpha^{r-1} n$,
  implying $k_r \ge 1 + \floor{(r-1)\log\invalpha}$ for this epoch, which in
  turn upper-bounds the number of stages of the epoch. Putting it all
  together, the total number of streaming passes, across all epochs, is
  \begin{align*}
      1 + \sum_{r=1}^{\ceil{\log_{1/\alpha} \Delta}} 
        \left( 3 \ceil{\frac{b}{k_r}} + 1 \right)
      &= O(\log\Delta) + O(b) \cdot \sum_{r=1}^{\ceil{\log_{1/\alpha} \Delta}} \frac{1}{k_r} \\
      &= O(\log\Delta) \cdot \sum_{r=1}^{\ceil{\log_{1/\alpha} \Delta}} \frac{1}{r} \\
      & = O(\log \Delta \cdot \log \log \Delta) \,. \qedhere
  \end{align*}
\end{proof}

This concludes the proof of our first major algorithmic result, which we now
recap.
\begin{theorem}[Restatement of \Cref{thm:det-multipass-coloring}]
  There is an efficient deterministic semi-streaming algorithm to
  $(\Delta+1)$-color an $n$-vertex graph, given a stream of its edges. The
  algorithm uses $O(n \log^2 n)$ bits of space and runs in $O(\log\Delta
  \log\log\Delta)$ passes. 
  \qed
\end{theorem}

\subsection{Extensions: List Coloring and Communication Complexity}

We can extend~\Cref{alg:determ} to handle the more general problem of
$(\deg+1)$-list-coloring. This requires a new technical lemma and a careful
refinement of some of the low-level details of the previous algorithm. 

  
\begin{theorem}[Restatement of \Cref{thm:d1lc-det-multipass}]
  Let $C$ be a set of colors of size $O(n^2)$. There is a deterministic
  semi-streaming algorithm for $(\text{degree}+1)$-list-coloring a graph $G$
  given a stream consisting of, in any order, the edges of $G$ and $(x,L_x)$
  pairs specifying the list $L_x$ of allowed colors for a vertex $x$,
  where $L_x \subseteq C$. The algorithm uses $O(n \log^2 n)$ bits of space
  and runs in $O(\log\Delta \log\log\Delta)$ passes.
\end{theorem}

Here is a technical lemma that is key to the proof of the above.

\begin{lemma}\label{lem:d1lc-partitioner}
  Let $s\ge 1$ be an integer, and let $C$ be a set. There exists a family $\cF$ of $O(|C|^2)$ partitions of $C$ so that, for every collection 
  $L_1,\ldots,L_t$ of subsets of $C$:
  \begin{align}
     \frac{1}{|\cF|} \sum_{\cR \in \cF} \sum_{i \in [t]} \max_{S \in \cR} (|L_i \cap S| - 1) \le \frac{1}{\sqrt{s}} \sum_{i\in[t]} (|L_i| - 1) \,, \label{eq:d1lc-subavg}
  \end{align}
  In particular, there must exist $\cQ \in \cF$ where $\sum_{i \in [t]} \max_{S \in \cQ} (|L_i \cap S| - 1)$ is less than the right hand side.
\end{lemma}

\begin{proof}
  Let $\cH$ be a 2-universal hash family $C \rightarrow [s]$, with $|\cH| = O(|C|^2)$. (For example, $\cH = \{(x \mapsto (a x + b \bmod p) \bmod s) : a,b \in \ZZ_p, a \ne 0 \}$ for $p$ prime and $\ge |C|$, as per \cite{CarterW79}.) Let $h$ be a randomly chosen element of $\cH$, and let $\cR = \{R_1,\ldots,R_s\}$ be the random partition for which $R_i = \{ x \in C : h(x) = i\}$. Consider the function $f(x) = x(x+1)/2$ defined on $[0,\infty)$; because it is convex and increasing on $[0,\infty)$, $f^{-1}(x) = \sqrt{2x + 1/4} - 1/2$ is concave and increasing on $[0,\infty)$.
  Because for all $z \ge 1$,
  $z - 1 = f^{-1}(\binom{z}{2})$, we have for any $i \in [t]$ that:
  \begin{align*}
    \max_{j \in [s]} (|L_i \cap R_j| - 1) \le f^{-1}\left(\max_{j \in [s]} \binom{L_i \cap R_j}{2}\right) \le f^{-1}\left( \sum_{j \in [s]} \binom{L_i \cap R_j}{2} \right) \,.
  \end{align*}
  Taking expectations and using the concavity of $f$ to apply Jensen's inequality:
  \begin{align*}
    \EE \max_{j \in [s]} (|L_i \cap R_j| - 1) &\le \EE f^{-1}\left( \sum_{j \in [s]} \binom{L_i \cap R_j}{2} \right)
        \le f^{-1}\left( \EE \sum_{j \in [s]} \binom{L_i \cap R_j}{2} \right) \,.
  \end{align*}
  Expressing the sum under the inverse function in terms of $h$ lets us apply the universality of the hash family:
  \begin{align*}
    \EE \sum_{j \in [s]} \binom{L_i \cap R_j}{2} = \EE \sum_{x,y \in L_i : x \ne y} \indic_{h(x) = h(y)} = \sum_{x,y \in L_i : x \ne y} \Pr[h(x) = h(y)] \le \binom{|L_i|}{2} \frac{1}{s} \,.
  \end{align*}
  We briefly detour to prove an inequality for $f$, holding for all $z \ge 1$:
  \begin{align*}
    f\left(\frac{1}{\sqrt{s}} (z - 1)\right) = \frac{\frac{1}{\sqrt{s}} (z - 1) \cdot (\frac{1}{\sqrt{s}} (z - 1) + 1)}{2}
        = \frac{1}{s} \frac{(z-1)(z + \sqrt{s} - 1)}{2} \ge \frac{1}{s} \binom{z}{2}\,,
  \end{align*}
  which implies $f^{-1}(\frac{1}{s} \binom{z}{2}) \le \frac{1}{\sqrt{s}} (z - 1)$. Thus:
  \begin{align*}
    \EE \max_{j \in [s]} (|L_i \cap R_j| - 1) \le f^{-1}\left(  \binom{|L_i|}{2} \frac{1}{s} \right) \le \sqrt{\frac{1}{s}} (|L_i| - 1) \,.
  \end{align*}
  By linearity of expectation, it follows
  \begin{align*}
    \EE \sum_{i \in [t]} \max_{j \in [s]} (|L_i \cap R_j| - 1) \le  \sqrt{\frac{1}{s}} \sum_{i \in [t]} (|L_i| - 1) \,.
  \end{align*}
  This is equivalent to Eq. \ref{eq:d1lc-subavg}, if we let $\cF$ be the set of possible values of $\cR$.
  
\end{proof}

\begin{proof}[Proof of \Cref{thm:d1lc-det-multipass}]
  
  There are two main changes to the algorithm in Theorem \ref{thm:det-multipass-coloring}. First, because the color lists $L_x$ are no longer fixed, computing $\slackwrt{P_x \cap Q}{x}$ for each $x \in U$ and $Q \in \cQ^{(i)}$ requires counting both $|\{y \in \Nhd(x) \setm U : \chi(x) \in (P_x \cap Q) \}|$ as before, and $|P_x \cap Q \cap L_x|$. As both quantities are integers in $[0,\ldots,\Delta + 1]$, and can be computed by incrementing counters each time an edge or (vertex, list of colors) pair arrives, the total space usage from this stage is still $O(\log \Delta) |U| 2^{k}$.
  
  The other change is that we now adaptively pick the sequence of partitions $\cQ^{(1)},\ldots,\cQ^{(\ell)}$, and use more stages.
  Instead of letting the number $\ell$ of stages be $\ceil{\log(\Delta+1)/k}$, we use $\ell = \ceil{2  \log(\Delta+1)/k} + 1$ stages instead. 
  For the first $\ceil{2 \log(\Delta+1)/k}$ stages, we adaptively construct partitions using
  Lemma \ref{lem:d1lc-partitioner} on the $L_x$ with $s$ set to $2^{k}$; the resulting
  partitions use $O(\ell \log |\cC|) = O(\log \Delta \log n)$ space to store in total.
  
  Finding the best partition from Lemma \ref{lem:d1lc-partitioner} is complicated
  by the fact that the algorithm can not exactly store the color lists $L_x$ for each vertex. Let $\cF$ be the family of partitions from  Lemma \ref{lem:d1lc-partitioner}.
  At the start of each stage, we use four passes over the stream to identify a partition $\cR \in \cF$ for which the quantity $\sum_{x \in U} a_{\cR}(P_x \cap L_x)$ is below average, for $a_{\cR}(S) = \max_{R \in \cR} (| S \cap R | - 1 )$. This can be done using the same method as was used to identify an approximately sub-average hash function $h^\star$ in \Cref{alg:determ}. In the first pass, we split $\cF$ into $O(|\cF|^{1/4})$ parts, and compute
  $\sum_{\cR} \sum_{x \in U} a_{\cR}(P_x \cap L_x)$ for each part; after the pass completes, we pick the part with the least value of this sum, split it into $O(|\cF|^{1/4})$ smaller parts,  and repeat the process. The fourth pass will
  compute $\sum_{x \in U} a_{\cR}(P_x \cap L_x)$ for individual partitions $\cR$ of the family $\cF$; we let $\cQ^{(i)}$ be the best partition from this pass.
  All this is possible because the value of $a_{\cR}(P_x \cap L_x)$ can be computed as soon as the pair $(x,L_x)$ arrives in the stream. Consequently, it is possible to compute, for any family $\cF$ of partitions, $\sum_{\cR \in \cF} \sum_{x \in U} a_{\cR}(P_x \cap L_x)$ in a single pass over the stream, using $O(\log n)$ bits of space. (These sums have integer values, so no approximation is necessary.) As $|\cH| = O(|C|^2) = O(n^4)$, each individual pass requires storing only $O(n \log n)$ bits worth of counters. 
  
  At the start of the first stage, since all $|L_x| \le \Delta + 1$, we have $ \sum_{x \in U} (|L_x \cap P_x| - 1) \le \Delta |U|$. Letting $j_x$ be the index of $P_{x,j} = P_x \cap Q^{(i)}_{j}$ chosen to succeed $P_x$, we have (due to Lemma \ref{lem:d1lc-partitioner}).
  \begin{align*}
    \sum_{x \in U} (|L_x \cap P_{x,j_x}| - 1) \le \sum_{x \in U} \max_{j \in [s]} (|L_x \cap P_x \cap Q^{(i)}_{j}| - 1) \le 2^{-k/2} \sum_{x \in U} (|L_x \cap P_x| - 1)
  \end{align*}
  Each stage reduces $\sum_{x \in U} (|L_x \cap P_x| - 1)$ by a factor of
  $2^{-k/2}$, so after $\ell - 1 = \ceil{2 \log(\Delta+1)/k}$ stages, we have
  \begin{align*}
    \sum_{x \in U} (|L_x \cap P_x| - 1) \le \Delta |U| (2^{- k/2 (\ell - 1)}) \le \frac{\Delta}{\Delta+1} |U| \le |U|
  \end{align*}
  In the last stage, we set $\cQ = \{ \{ x \} : x \in \cC$, where $\cC =
  \bigcup_{x\in U} L_x$. Unlike the other stages, where $|\cQ| \le 2^{k}$, we need to run an additional pass to record, for each $x \in U$, the values of $|L_x \cap P_x|$. 
  This requires only $O(|U| \log n)$ bits. 
  In the following pass to compute $\slackwrt{P_x \cap Q}{x}$ for each $x \in U$ and $Q \in \cQ$, 
  we use the fact that $\slackwrt{P_x \cap Q}{x}$ will only be one if $Q \subseteq
  P_x \cap L_x$ and there is no $y \in \Nhd(x) \setm U$ satisfying $\chi(y)
  \in Q$ to save space; instead of tracking sums for every $(x,Q) \in U \times
  \cC$ combination, we store a $\{0,1\}$ value for each $(x,Q) \in \sqcup_{x
  \in U} \{ (x,Q) : Q \in L_x \cap P_x \}$ which is initialized to $1$ and set
  to $0$ if the stream contains an edge to a neighboring $y \in [n] \setm
  U$ with color in $Q$.  After this stage, the condition $|L_x| \le 1$ holds, as required for the proof of Theorem \ref{thm:det-multipass-coloring} to work.
  
  Despite the less efficient partitioning scheme, the algorithm still uses
  roughly the same amount of space; for all but the last stage, it still uses
  $2^{k} |U|$ counters. The last stage requires one bit for each element in a list $L_x$ -- but since $\sum_{x \in U} (|L_x| - 1) \le |U|$, we have $\sum_{x \in U} |L_x| \le 2 |U|$, which implies only $2 |U|$ bits are needed.
  
  Storing the per vertex partitions $P_x$ requires only $\ell k + \log(|\cC|) = O(\log n)$ bits, each, at a given point in the algorithm. As in the original algorithm, each partition $P_x$ can be determined using the sequence of sets from $\cQ^{(1)},\ldots,\cQ^{(\ell)}$ that contain it.
  
  The analysis to prove that the potential does not increase by much requires no adjustment.
\end{proof}

Finally, we record the following corollary of the above algorithms on the communication complexity of $(\Delta+1)$ coloring that may be of independent interest.  

\begin{corollary}\label{thm:delta-plus-1-communication}
 There is a communication protocol for finding a $(\Delta+1)$ coloring of any input graph whose edges are partitioned between two players using $O(n\log^4{n})$ bits of communication and $O(\log \Delta \log \log \Delta)$ rounds of communication. 
\end{corollary}

\begin{proof}
This follows from a standard reduction from a streaming algorithm to a communication protocol.

Let Alice and Bob be the two players, who receive disjoint sets of edges $A$ and $B$, respectively. They will run \Cref{alg:determ} on the stream whose first half contains the edges of $A$, and whose second half contains the edges of $B$. To do this, Alice initializes the streaming algorithm, and runs it on the first half of the stream. She then sends a message encoding the state of the algorithm to Bob, who decodes the message and runs the algorithm on the second half of the stream. Bob then sends the updated state of the streaming algorithm back to Alice. This process is repeated once for each pass of the streaming algorithm; since the algorithm uses $O(n \log^2 n)$ bits of space, uses $O(\log \Delta \log \log \Delta) = O(\log^2 n)$ passes, the total number of bits sent by this protocol is $O(n \log^4 n)$.
\end{proof}

While it is not hard to obtain an $O(n \cdot \polylog{(n)})$ communication protocol for $(\Delta+1)$ coloring by simulating the greedy algorithm (and running binary search between Alice and Bob for finding an available color for each vertex), the interesting part of~\Cref{thm:delta-plus-1-communication} is that we can achieve a similar communication guarantee in a much smaller number of rounds of communication.

%% file: robust.tex

\section{Coloring Robustly Against an Adaptive Adversary}\label{sec:coloring-robustly}

We now turn to the adversarially robust streaming setting. As a reminder, this
is inherently a single-pass setting and our algorithms are allowed to use
randomness. However, an algorithm needs to be correct against an adversary who
constructs the input graph adaptively by inserting upcoming edges based on the
colorings returned by the algorithm. This means that the stream elements might
depend on past outputs, which in turn depend on the random bits used by the
algorithm. While \cite{AssadiCK19} gave a semi-streaming $(\Delta+1)$-coloring
algorithm in the ``non-robust'' oblivious adversary setting where the stream
is fixed in advance, \cite{ChakrabartiGS22} showed that a robust
semi-streaming algorithm must use $\Omega(\Delta^2)$ colors. They also gave an
$O(\Delta^3)$-coloring robust algorithm. In this section, we give an improved
$O(\Delta^{5/2})$-coloring algorithm. 

We assume that $\sqrt{\Delta}$ is an integer (if not, we can work with $\left\lceil\sqrt{\Delta}\right\rceil$ which will not affect the asymptotic color or space bounds). We also assume that $\Delta = \Omega(\log^2 n)$; if $\Delta$ is smaller, we can store the entire graph in semi-streaming space and then color it optimally.

The following graph-theoretic concept plays a crucial role in our algorithm.

\begin{definition}[degeneracy]
  The \emph{degeneracy} of a graph $G$ is the least integer value $\kappa$ for
  which every induced subgraph of $G$ has a vertex of degree $\le \kappa$.
  Equivalently, it is the least value $\kappa$ for which there is an acyclic
  orientation of the graph where the maximum out-degree of any vertex is $\le
  \kappa$. By greedily assigning colors to the vertices of this orientation of
  $G$ in reverse topological order, one obtains a proper $(\kappa+1)$-coloring
  of $G$; we refer to this as a $(\text{degeneracy}+1)$-coloring.
\end{definition}

\subsection{High-Level Description and Techniques}

We first set up some terminology to help us outline our algorithm.

  \begin{itemize}[topsep=2pt,itemsep=1pt]
      \item \textbf{Buffer.} As the stream arrives, we explicitly store a buffer $B$ of at most $n$ edges. When the buffer is full (i.e., has reached its capacity of $n$ edges), we empty it completely, and move on to storing the next batch of $n$ edges. 
      
      \item \textbf{Epoch.} We say we are in the $i$th epoch when we are storing the $i$th chunk of $n$ edges in our buffer.
      
      \item \textbf{Level.} We define levels for the vertices with respect to their degree in the (entire) graph seen so far. At the point of query, we say that a vertex is in level $\ell$, if its degree in the current graph is in $((\ell-1)\sqrt{\Delta},\ell \sqrt{\Delta}]$. 
      
      \item \textbf{Zone (fast and slow).} We define zones (fast or slow) for
      the vertices with respect to their degree in the \emph{buffer} $B$. At
      the time of query, we say that a vertex $v$ is in the \emph{fast} zone
      if $\deg_B(v) >\sqrt{\Delta}$; otherwise, we say that it is in the
      \emph{slow} zone. We also use the terms \emph{slow vertex} and
      \emph{fast vertex}, respectively.
      
      \item \textbf{Block.} We have multiple coloring functions, denoted by $h_i$ and $g_i$, that assign each node a color uniformly at random from a palette of suitable size (not to be confused with the final proper coloring; these colorings are improper). As a result, we obtain a partition of the nodes into monochromatic classes that we call ``blocks.'' A block produced by a coloring function $f$ is called an $f$-block. More formally, for each $c$ in the range of $f$, the set of nodes $\{v\in V: f(v)=c\}$ is called an $f$-block. 
      
     \item \textbf{$f$-Monochromatic.} An edge $\{u,v\}$ with $f(u)=f(v)$ is called $f$-monochromatic.

     \item \textbf{$f$-Sketches.} For a function $f$ we call the underlying sketch of the algorithm, which receives edges of the graph and stores it only if it is $f$-monochromatic, as an $f$-sketch.
      
  \end{itemize}

Next, we describe how to color the slow vertices using $O(\Delta^{5/2})$
colors in semi-streaming space. Then we do the same for the fast vertices.

\paragraph{Coloring slow vertices} Consider breaking the edge stream into
$\Delta$ ``chunks'' of size $n$ each. As described above, our buffer $B$
basically stores a chunk from start to end, and then deletes it entirely and
moves on to the next chunk. We initialize $\Delta$ many coloring functions
$h_1, \ldots, h_\Delta$ that run in parallel. For each $i$, the function $h_i$
assigns each node a color from $[\Delta^2]$ uniformly at random.  An
$h_i$-sketch (see definition above) processes the prefix of the stream until
the end of chunk $i$. Recall that ``processing'' means it stores a received
edge $(u,v)$ in the set $A_i$ if it is $h_i$-monochromatic. 

Suppose a query arrives in the current epoch $\curr$. Fix a subgraph induced by \emph{only} the slow vertices in an arbitrary $h_\curr$-block on the edge set $A_{\curr-1}\cup B$ (set $A_0:=\emptyset$). Recolor this subgraph using an offline $\Delta'+1$-coloring algorithm where $\Delta'$ is its max-degree. Now do this for each $h_\curr$-block, using fresh palettes for the distinct blocks. We then return the resultant coloring (for the slow nodes). We now argue that the number of edges stored in $(\cup_i A_i)$ is roughly $O(n)$ and the number of colors used is $O(\Delta^{5/2})$.

Observe that for each $i$, the $h_i$-sketch processes the prefix of the stream until the end of epoch $i$. But note that, until that point, we only base our output on $A_j$s for $j<i$, which are independent on $h_i$ in particular. Therefore, we ensure that each $h_i$-sketch processes a part of the stream independent of their randomness. Hence, an edge $(u,v)$ received by an $h_i$-sketch is $h_i$-monochromatic with probability $1/\Delta^2$. Since it receives at most $n\Delta$ edges, it {\em stores} only $O(n\Delta/\Delta^2)=O(n/\Delta)$ edges in expectation in $A_i$. By a Chernoff Bound argument, the actual value is tightly concentrated around this expectation w.h.p. Then, the $\Delta$ sets $A_1,\ldots,A_\Delta$ store roughly $O(n/\Delta\cdot \Delta)=O(n)$ edges in total w.h.p. 

Now, we first verify that it properly colors the graph induced by the slow nodes. Observe that we indeed stored each edge of the input graph, which is contained in any $h_\curr$ block of slow vertices, in $A_{\curr-1}\cup B$. This is because if it is in $B$, we have definitely stored it, and otherwise, it was in an epoch $\leq \curr-1$. Therefore, the $h_{\curr-1}$-sketch received it and must have stored it in $A_{\curr-1}$. This means each intra-block edge is properly colored by the offline algorithm, and each inter-block edge is also properly colored since we use distinct palettes for distinct blocks.

Now we argue the color bound. For each slow node, an $h_i$-sketch receives at most $\Delta$ edges incident to it and hence, $A_i$ stores $O(\Delta\cdot 1/\Delta^2)= O(1/\Delta)$ edges incident to it in expectation (by the previous argument). By a Chernoff Bound argument and taking union bound over all nodes, we get that each of them has degree roughly $O(\log n)$ in $A_i$ w.h.p. Further, since these nodes are slow, they have degree at most $\sqrt{\Delta}$ in $B$. Therefore, the degree of each slow node in the edge set $A_{\curr-1}\cup B$ is $O(\sqrt{\Delta}+\log n)=O(\sqrt{\Delta})$ since $\Delta$ is assumed to be $\Omega(\log^2 n)$. Hence, each $h_\curr$-block of slow nodes induced on $A_{\curr-1}\cup B$ is colored with a fresh palette of $O(\sqrt{\Delta})$ colors by the offline algorithm. There are $\Delta^2$ many $h_\curr$-blocks, and therefore, we use $O(\Delta^2\cdot \sqrt{\Delta})=O(\Delta^{5/2})$ colors. \smallskip

\paragraph{Coloring fast vertices} To handle these, we use another
$\sqrt{\Delta}$ coloring functions $g_1,\ldots,g_{\sqrt{\Delta}}$. Each $g_i$
assigns each node a color from $[\Delta^{3/2}]$ uniformly at random. When an
edge $\{u,v\}$ arrives, let $\ell$ be the maximum between the two levels of
$u$ and $v$. We send it to the $g_i$-sketches for all $i\geq \ell+1$. Recall
that a $g_i$-sketch then stores the edge in the set $C_i$ only if it is
$g_i$-monochromatic, i.e., if $g_i(u)=g_i(v)$. 

We prove that each $g_i$-sketch processes edges independent of their randomness. This is the tricky part. Intuitively,
for each edge $\{u,v\}$ that a $g_i$-sketch receives, the levels of $u$ and $v$ were strictly smaller than $i$ when it
was inserted. Thus, the only values $g_j(u)$ and $g_j(v)$ that were used to return outputs until then were for $j<i$.
Hence, $g_i(u)$ and $g_i(v)$ were completely unknown to the adversary when $\{u,v\}$ was inserted. Thus, the edge
stream received by each $g_i$-sketch is independent of the randomness ``that matters'' in processing it. Hence, since
the probability that each edge is $g_i$-monochromatic is $1/\Delta^{3/2}$, each $g_i$-sketch stores roughly
$O(n\Delta/\Delta^{3/2})=O(n/\sqrt{\Delta})$ edges in $C_i$. Thus, the total number of edges stored by
$C_1,\ldots,C_{\sqrt{\Delta}}$ is $O(n/\sqrt{\Delta}\cdot \sqrt{\Delta})=O(n)$.

When a query arrives, for each level $i$, we consider the fast vertices in each $g_i$-block. Then consider the subgraph induced by these vertices on the edge set $C_i\cup B$. Color it using a (degeneracy+1)-coloring offline algorithm. We prove that this colors the fast vertices properly with $O(\Delta^{5/2})$ colors. 

To verify that it is a proper coloring, we need to show that the subgraph of $G$ induced on each $g_i$-block of fast vertices is stored in $C_i\cup B$. This follows from the ``fastness'' property of the nodes: if any such edge $\{u,v\}$ is not in the buffer $B$, then, since the degrees of $u$ and $v$ increased by at least $\sqrt{\Delta}$ in the buffer, the nodes $u$ and $v$ must have been at levels lower than $i$ when $\{u,v\}$ was inserted. Therefore, it was fed to the $g_i$-sketch, which stored it since it is $g_i$-monochromatic. Hence, each intra-block edge of fast vertices is properly colored by the offline algorithm, and each inter-block edge is also properly colored since we use distinct palettes for distinct blocks.

\subsection{The Robust Algorithm and its Analysis}

We now present the pseudocode of our algorithm in \Cref{alg:robust}. The analysis of correctness, robustness, space usage, and color bound is given below. 

 \begin{algorithm}[ht!]
  \caption{Adversarially Robust $O(\Delta^{2.5})$-Coloring in Semi-Streaming Space
    \label{alg:robust}}
  \begin{algorithmic}[1]
  \Statex \textbf{Input}: Stream of edge insertions of an $n$-vertex graph $G=(V,E)$
  \Statex
    \Statex \underline{\textbf{Initialize}:}
    \State $d(v) \gets 0$ for each $v\in V$ \Comment{degree counters}
    \For{$i$ from $1$ to $[\Delta]$}  \Comment{$\Delta$ parallel copies for $\Delta$ possible epochs}
    \State Let $h_i : V \rightarrow \left[\Delta^2\right]$ be uniformly random  \Comment{$h_i$ assigns each node a color from $\left[\Delta^2\right]$ u.a.r.}
    \State $A_i \gets \emptyset$ \Comment{edges stored by $h_i$-sketch}
    \EndFor

     \For{$i$ from $1$ to $\left[\sqrt{\Delta}\right]$} \Comment{$\sqrt{\Delta}$ parallel copies for $\sqrt{\Delta}$ possible levels}
    \State Let $g_i : V \rightarrow \left[\Delta^{3/2}\right]$ be uniformly random \Comment{$g_i$ assigns each node a color from $\left[\Delta^{3/2}\right]$ u.a.r.}
    \State $C_i \gets \emptyset$ \Comment{edges stored by $g_i$-sketch}
    \EndFor

      \State $B\gets \emptyset$ \Comment{buffer}
      \State ${\curr} \gets 1$ \Comment{current epoch number}
  \Statex 
  \Statex\underline{\textbf{Process}(edge $\{u,v\}$):}
  \If{$|B| = n$} \label{step:buffer-replace-check}
      \State $B \gets \emptyset$; $\curr\gets\curr+1$ \Comment{Empty buffer if full and update epoch number}
      \EndIf
  \State $B \gets B \cup \{\{u,v\}\}$; 
 \Comment{Update buffer and buffer size}
            \State $d(u) \gets d(u) + 1$; $d(v) \gets d(v) + 1$ \label{step:deg-increase} \Comment{Increase degrees of $u$ and $v$}
      \For{$i$ from  $({\curr}+1)$ to $\Delta$} \Comment{Consider copies corresponding to higher epochs}
      \If{$h_i(u) = h_i(v)$} $A_i \gets A_i \cup \{\{u,v\}\}$ \Comment{Store $h_i$-monochromatic edges in $A_i$} \label{step:slow-record} 
      \EndIf
      \EndFor
      \For{$i$ from $\left\lceil\frac{\max\{d(u),d(v)\}}  {\sqrt{\Delta}}\right\rceil + 1 $ to $\Delta$} \Comment{Consider levels higher than both $u$ and $v$}
      \If{$g_i(u) = g_i(v)$} 
      $C_i \gets C_i \cup \{\{u,v\}\}$. \Comment{Store $g_i$-monochromatic edges in $C_i$} \label{step:fast-record}
      \EndIf
      \EndFor
      
      \Statex
      \Statex \underline{\textbf{Query}():}
      \State $F \gets \{ v \in V : \deg_B(v) > \sqrt{\Delta} \}$ \Comment{$F$ contains fast vertices that have received $>\sqrt{\Delta}$ edges in the buffer}
      \State $S \gets V \setm F$ \Comment{$S$ contains the remaining slow vertices}
       \For{$c$ from $1$ to $[\Delta^2]$} \label{step:slow-color} 
       \State $S_{\curr}(c)\gets \left\{w \in S: h_{\curr}(w) = c\right\}$ \Comment{Consider each $h_\curr$-block among slow vertices}
      \State Using fresh colors, (degree+1)-color subgraph induced by $S_{\curr}(c)$ on edge set $A_{\curr-1}\cup B$
      \EndFor
      \For{$\ell$ from $1$ to $\left[\sqrt{\Delta}\right]$}  \label{step:fast-color}
      \For{$c$ from $1$ to $\left[\Delta^{3/2}\right]$}
       \State $F(\ell,c)\gets \left\{w \in F: \left\lceil  \frac{d(w)}{\sqrt{\Delta}}\right\rceil=\ell, \text{ and } g_{\ell}(w) = c\right\}$  \Comment{Consider each $g_\ell$-block among fast vertices}
      \State Using fresh colors, (degeneracy+1)-color subgraph induced by $F(\ell,c)$ on edge set $C_{\ell}\cup B$
      \EndFor
      \EndFor
      \State Output resultant coloring for $S\cup F=V$
      
  \end{algorithmic}
\end{algorithm}

\begin{lemma}\label{lem:c-degree-limit}
  With high probability, for all vertices $x \in V$, we have $\sum_{i\in [\sqrt{\Delta}]} d_{C_i}(v) = O(\log n)$.
\end{lemma}

\begin{proof}
  For any $x \in V$, let $D$ be the random variable for the degree of $x$ at the end of the stream, and let $\{x,Y_1\}$, $\{x,Y_2\}$, ... $\{x,Y_D \}$ be the edges added adjacent to $x$ by the adversary, in order. For all $k \in [\Delta]$ and $\ell \in \sqrt{\Delta}$, let $Z_{k, \ell}$ be the random variable which is 1 if $k \le D$ and the algorithm stores the edge $\{x,Y_k\}$ in the set $C_\ell$, and zero otherwise. The edge $\{x,Y_k\}$, assuming it exists, will be stored in $C_i$ only if $g_i(x) = g_i(Y)$ and $i \ge \ceil{\frac{\max(d(x),d(Y_k))}{\sqrt{\Delta}}} + 1$, where $d(x)$ and $d(Y_k)$ are the values of the degree counter at the time the edge was added. Now, consider the sequence of random variables,
  \begin{align}
    Z_{1,1},\ldots,Z_{1,\sqrt{\Delta}}, Z_{2,1},\ldots,Z_{2,\sqrt{\Delta}}, \ldots, Z_{\Delta,1},\ldots,Z_{\Delta,\sqrt{\Delta}} \label{eq:rvsequence}
  \end{align}
  Their sum is precisely $\sum_{\ell \in [\sqrt{\Delta}]} d_{C_\ell }(x)$. In order to bound this sum with high probability, we would like to use \Cref{lem:forward-concentration}, but in order for that to work we need to prove that the expectation of a given $Z_{k,\ell}$, conditional on all the earlier terms in the sequence, is bounded. Let $\prec$ indicate the lexicographic order on pairs of the form $(k',\ell')$, so that $(k'',\ell'') \prec (k', \ell')$ if either $k'' < k'$, or ($k'' = k'$ and $\ell'' < \ell'$.). Define $Z_{\prec (k,\ell)}$ to be the vector  $(Z_{k',\ell'})_{(k',\ell') \prec (k, \ell)}$. We want to prove an upper bound on $\EE[Z_{k,\ell} \mid Z_{\prec (k,\ell)}]$. Intuitively, the edge $\{x,Y_k\}$ chosen by the adversary will either definitely not be stored in $C_k$ -- because e.g. one of the degrees of the endpoints is too large -- or, when it is time to check whether $g_\ell(x) = g_\ell(Y_k)$, the value read from $g_\ell(Y_k)$ will not have been revealed to the adversary so far, nor will it have been read as part of any test to determine if $\{x,Y_{k'}\}$ should be stored in $C_{k'}$, for $(k',k') \prec (k,k)$; so $g_\ell(Y_k)$ will be independent of the variables in $Z_{\prec (k,\ell)}$, and will equal $g_\ell(x)$ with probability exactly $1/\Delta^{3/2}$. Either way, we will find $\EE[Z_{k,\ell} \mid Z_{\prec (k,\ell)}] \le 1/\Delta^{3/2}$. (A more formal proof of this fact is provided as \Cref{lem:c-degree-limit-formal} in \Cref{sec:deferred-proofs})
  
  Now, applying \Cref{lem:forward-concentration} to the sequence of random variables from Eq. \ref{eq:rvsequence}, we obtain:
  \[
     \Pr\left[ \sum_{\ell \in [\sqrt{\Delta}]} d_{C_\ell}(x) \ge 5 \log n \right] 
     \le \Pr\left[ \sum_{k \in [\Delta]} \sum_{\ell \in [\sqrt{\Delta}]} Z_{k,\ell} \ge \Delta^{3/2} \cdot \frac{1}{\Delta^{3/2}} (1 + 4 \log n) \right]
     \le 2^{- 4 \log n} = \frac{1}{n^4} \,.
  \]
  
  Then taking a union bound of this event for each $x \in V$, we conclude that $\sum_{i\in [\sqrt{\Delta}]} d_{C_i}(x) = O(\log n)$ holds for \emph{all} $x$ with high probability.
\end{proof}

\begin{lemma}\label{lem:a-degree-limit}
   With high probability, for all vertices $x \in V$, we have $\sum_{i\in [\Delta]} d_{A_i}(v) = O(\log n)$.
\end{lemma}

\begin{proof}
  The argument here is essentially the same as for the proof of \Cref{lem:c-degree-limit}, so we will skip most of the details, and describe briefly what changes.
  
  Instead of defining indicator random variables $Z_{k,\ell}$ for the event that
  the \Cref{alg:robust} stores a given edge $\{x,Y_k\}$ in $C_\ell$, we define indicator random variables $Z_{k,i}$, for $i \in [\Delta]$, for the event that 
  the algorithm stores $\{x,Y_k\}$ in $A_i$. With a similar lexicographically ordered sequence of the $Z_{k,i}$, one can prove that each random variable $Z_{k,i}$ has expectation $\le \frac{1}{\Delta^2}$, even after conditioning on the values of all the earlier random variables in the sequence. This will use the observation that, if the answer to whether the edge $\{x,Y_k\}$ will be stored in the set $A_i$ depends on the value of $h_i(Y_k)$, then the value of $h_i(Y_k)$ has not been revealed to the adversary. Applying \Cref{lem:forward-concentration}, one will then find:
  \[
     \Pr\left[ \sum_{i\in [\Delta]} d_{A_i}(x) \ge 5 \log n \right] 
     \le \Pr\left[ \sum_{k \in [\Delta]} \sum_{i \in [\Delta]} Z_{k,i} \ge \Delta^{2} \cdot \frac{1}{\Delta^{2}} (1 + 4 \log n) \right]
     \le 2^{- 4 \log n} = \frac{1}{n^4} \,.
  \]
  The proof is completed by taking a union bound.
\end{proof}

\begin{lemma}\label{lem:space-usage}
  The space usage of \Cref{alg:robust} is $\tO(n)$ bits, with high probability.
\end{lemma}

\begin{proof}
  By Lemmas \ref{lem:a-degree-limit} and  \ref{lem:c-degree-limit}, all vertices $x \in V$ satisfies $\sum_{i\in [\Delta]} d_{C_i}(v) = O(\log n)$, and $\sum_{i\in [\Delta]} d_{C_i}(v) = O(\log n)$, with high probability.  Since $|C_i| = \frac{1}{2} \sum_{x \in V} d_{C_i} (x)$, and $|A_i| = \frac{1}{2} \sum_{x \in V} d_{A_i} (x)$, it follows \Cref{alg:robust} stores $O(n \log n)$ edges in total in $\bigcup_{i \in [\Delta]} A_i \cup \bigcup_{i \in [\sqrt{\Delta}]} C_i$.
  Additionally, it stores a buffer $B$ of $n$ edges. Hence, the algorithm stores $\tO(n)$ edges in total. Further, it stores a degree counter for each node and a couple of counters for tracking the buffer size and the epoch number. These take an additional $\tO(n)$ bits of space. Thus, the total space usage of the algorithm is $\tO(n)$ bits.
\end{proof}

\begin{lemma}\label{lem:advrob-degeneracy}
  At any point in the stream, for each $\ell\in [\sqrt{\Delta}]$ and $c\in [\Delta^{3/2}]$, the degeneracy of the subgraph induced by the vertex set $F(\ell,c)$ on the edge set $C_\ell\cup B$ is $O(\sqrt{\Delta})$, w.h.p.
\end{lemma}

\begin{proof}
  To every vertex $v \in F(\ell, c)$, define $t_v$ to be the length of the input stream at the time that the degree counter $d(v)$ of $v$ increased from $(\ell - 1) \sqrt{\Delta}$ to $(\ell - 1) \sqrt{\Delta} + 1$; in other words, the time that vertex $v$ entered level $\ell$. By \Cref{lem:c-degree-limit}, with high probability it holds that $d_{C_\ell}(v) = O(\log n)$, so the set $C_\ell$ contributes at most $O(\log n) = O(\sqrt{\Delta})$ to the degeneracy of the induced subgraph of the edge set $C_\ell\cup B$ on the vertex set $F(\ell, c)$.
  
  It thus suffices to prove that the degeneracy of the graph $H$ on vertices of $F(\ell, c)$ formed by edges from $B \setm C_\ell$ is $\le \sqrt{\Delta}$. Orient each edge $\{u,v\}$ in $H$ from $u$ to $v$ if $t_v \ge t_u$, and from $v$ to $u$ otherwise. We will prove that the out-degree of each vertex from $F(\ell, c)$ in $H$ will be $\le \sqrt{\Delta}$. 
  
  Fix some $x \in F(\ell, c)$; for each edge $(x,y) \in H$, let $d_{xy}$ be the value of $d(x)$ directly after the streaming algorithm processed the edge $\{x,y\}$. Because $x \in F(\ell, c)$, we have $d_{x y} \le \ell \sqrt{\Delta}$. Since $x,y \in F(\ell, c)$, $g_\ell(x) = g_\ell(y) = c$. Because $\{x,y\} \in B \setm C_\ell$, $\max(d_{x y}, d_{y x})$ must have been $\ge (\ell - 1) \sqrt{\Delta} + 1$ -- otherwise the algorithm would have recorded the edge $\{x,y\}$ in $C_\ell$. Now the orientation of the edge applies: because $t_y \ge t_x$,
  the vertex $x$ must have reached degree $ (\ell - 1) \sqrt{\Delta} + 1$ at the same time or before $y$ did. Thus $d_{x y} \le (\ell - 1) \sqrt{\Delta}$ implies $d_{y x} \le (\ell - 1) \sqrt{\Delta}$; since we know $\max(d_{x y}, d_{y x}) > (\ell - 1) \sqrt{\Delta}$, it follows $d_{x y} \le (\ell - 1) \sqrt{\Delta}$. Since the variable $d(x)$ increases with each new edge adjacent to $x$ that arrives, and $d_{x y}  \in [(\ell - 1) \sqrt{\Delta} + 1, \ell  \sqrt{\Delta}]$ for all out-edges $(x,y)$ of $x$ in $H$, we conclude by the pigeonhole principle that $x$ has out-degree $\le \sqrt{\Delta}$ in $H$. This completes the proof that the degeneracy of $H$ is $\sqrt{\Delta}$, and thus of the lemma.

\end{proof}

\begin{lemma}
  Whenever queried, \Cref{alg:robust} outputs a proper coloring of the current graph $G$ and uses $O(\Delta^{5/2})$ colors w.h.p. 
\end{lemma}

\begin{proof}
  By \Cref{lem:a-degree-limit} and \Cref{lem:c-degree-limit}, with high probability,
  \begin{align}
      \max_{x \in V} \left( \sum_{i\in [\sqrt{\Delta}]} d_{C_i}(v) + \sum_{i\in [\Delta]} d_{A_i}(v) \right) = O(\log n) \label{eq:coloring-assumption}
  \end{align} We shall see that if this holds, then \Cref{alg:robust} will produce an $O(\Delta^{2.5})$ coloring of the graph.
  
  The total number of colors used is the sum of the number of colors used for the coloring of each of the subsets of vertices $S_\curr(c)$ (for $c \in [\Delta^2]$) and $F(\ell,c)$ (for $c \in \Delta^{3/2}, \ell \in \sqrt{\Delta}$). Because each of these subsets uses a fresh set of colors, and the subsets together disjointly cover the entire vertex set, the coloring output by \Cref{alg:robust} is valid if an only if all the individual subset colorings are valid.
  
  For each $c \in [\Delta^2]$, consider the set $S_\curr(c)$. For each edge $\{x,y\}$ in the graph, both of whose endpoints are in $S_\curr(c)$, we observe that either the edge  $\{x,y\}$ was added while the value of $\curr$ was less than it was now -- in which the algorithm would have stored $\{x,y\} \in A_\curr$, because $h_\curr(x) = h_\curr(y)$ -- or edge $\{x,y\}$ was added while $\curr$ had its current value -- in which case $\{x,y\}$ is in the set $B$. Thus, $A_\curr \cup B$ includes all the edges of the subgraph of $G$ induced by $S_\curr(c)$, so the degree + 1 coloring of $S_\curr(c)$ will be valid.
  
  Every vertex $x$ in $S_\curr(c)$ satisfies $\deg_{B}(x) \le \sqrt{\Delta}$, by the definition of the set $S$ of slow vertices. By Eq. \ref{eq:coloring-assumption}, $\deg_{C_\curr}(x) = O(\log n) = O(\sqrt{\Delta})$. Thus the maximum degree the edge set $C_\curr \cup B$ for any vertex in $S_\curr(c)$ will be $O(\sqrt{\Delta})$, and so a degree+1 coloring will only use $O(\sqrt{\Delta})$ colors.
  
  Now for $c \in [\Delta^{3/2}]$ and $\ell \in [\sqrt{\Delta}]$, consider the set $F(\ell, c)$ of vertices. To prove that the coloring of this set is correct, we must show that every edge $\{x,y\}$ which is contained in $G$, and which has both endpoints in $F(\ell, c)$, must be recorded in either $B$ or in $C_\ell$. Let $d_x$ and $d_y$ be the values of $d(x)$ and $d(y)$ after the \Cref{alg:robust} processes the edge $\{x,y\}$, i.e., after Line  \ref{step:deg-increase} has executed. We have two cases: either $\ell_{x,y} = \ceil{ \max(d_x,d_y)/\sqrt{\Delta} }$ is equal to $\ell$, or it must be less than $\ell$. If $\ell_{x,y} < \ell$, then the edge will be recorded in $C_\ell$ by Line \ref{step:fast-record}. Both the degree check and the check that  $g_\ell(x) = g_\ell(y)$ will pass, the latter because $x,y \in F(\ell, c)$ implies $g_\ell(x) = g_\ell(y) = c$. On the other hand, if $\ell_{x,y} = \ell$, then say without loss of generality that $\ceil{ d_x/\sqrt{\Delta} } = \ell$ -- this implies the degree of $x$ just after the edge $\{x,y\}$ was added was at least $(\ell - 1) \sqrt{\Delta} + 1$. Meanwhile, because $x \in F(\ell, c)$, the current degree of $x$ must be at most $(\ell - 1) \sqrt{\Delta}$. As each new edge adjacent to $x$ increases $d(x)$ by one, $\{x,y\}$ must be one of the $\sqrt{\Delta}$ most recent edges added adjacent to $x$. Since $x \in F$, the last $\sqrt{\Delta}$ edges adjacent to $x$ are all stored in $B$, and thus $\{x,y\} \in B$. The completes the proof that the coloring of $F(\ell, c)$ will be correct.
  
  By \Cref{lem:advrob-degeneracy}, the degeneracy of the subgraph induced by the vertex set $F(\ell, c)$ on edge set $C_\ell \cup B$ will be $O(\sqrt{\Delta})$, assuming Eq. \ref{eq:coloring-assumption} holds. As \Cref{alg:robust} computes a degeneracy+1 coloring of this graph, it will use $O(\sqrt{\Delta})$ colors.
  
  We have proven that each of the subsets of the form $S_\curr(c)$ or $F(\ell, c)$ will be properly colored using $O(\sqrt{\Delta})$ fresh colors. Since there are $2 \Delta^2$ such subsets in total, we conclude that algorithm \Cref{alg:robust} produces an $O(\Delta^{5/2})$ coloring of the graph as a whole.
\end{proof}

\begin{corollary}\label{cor:adv-space-color-tradeoff}
   By adjusting parameters of \Cref{alg:robust}, we can obtain a robust $O(\Delta^{(5 - 3\beta)/2})$-coloring algorithm using $O(n \Delta^\beta)$ space.
\end{corollary}
\begin{proof}
   These parameter changes do not significantly affect the proofs of correctness for \Cref{alg:robust}. 

   As before, we assume that the powers of $\Delta$ given here are integers, and that $\Delta = \Omega(\log^2 n)$:
   \begin{itemize}
       \item Change the buffer replacement frequency (Line \ref{step:buffer-replace-check}) from $n$ to $n \Delta^\beta$. Because a graph stream with maximum degree $\Delta$ contains at most $n \Delta / 2$ edges, reduce the number of epochs from $\Delta$ to  $\Delta^{1-\beta}$. The \texttt{for} loops initializing, updating, and querying the variables $h_i$ and $A_i$ should have bounds adjusted accordingly.
       
       \item Reduce the range of the functions $h_i$ from $[\Delta^2]$ to $[\Delta^{2-2 \beta}]$. The expected number of edges stored in all of the sets $A_i$ will now be roughly:
       \begin{align*}
           \frac{\text{\# epochs} \times |G|}{\text{\# slow blocks}} = \frac{\Delta^{1 - \beta } \cdot O(n \Delta)}{\Delta^{2 - 2 \beta} } =   O(n \Delta^\beta) \,,
       \end{align*}
       and with high probability, the space usage should not exceed this by more than a logarithmic factor.
       
       \item Increase the threshold for a vertex to be considered "fast" from $\sqrt{\Delta}$ to $\Delta^{(1+\beta)/ 2}$. To match this, the level of a vertex will now be computed as $\ceil{ \frac{d(v)}{ \Delta^{(1+\beta)/ 2 }} }$, and the number of levels reduced from $\sqrt{\Delta}$ to  $\Delta^{(1-\beta)/ 2}$. Again, all of the \texttt{for} loops related to the fast zone of the algorithm need to have their bounds adjusted.
       
       \item Reduce the range of the functions $g_\ell$ from $[\Delta^{3/2}]$ to 
       $[\Delta^{(1-\beta) 3 / 2}]$. The expected number of edges stored in all of the sets $C_\ell$ will now be roughly:
       \begin{align*}
           \frac{\text{\# levels} \times |G|}{\text{\# fast blocks}} = \frac{\Delta^{(1-\beta)/ 2} \cdot O(n \Delta)}{\Delta^{(1-\beta) 3 / 2} } =   O(n \Delta^{2\beta}) \,.
       \end{align*}

   \end{itemize}
   
   The number of colors used by the vertices in the slow zone will be:
    \begin{align*}
       \text{\# slow blocks} \times (O(\text{\# fast threshold}) + O(\log n)) = \Delta^{2-\beta} O(\Delta^{(1+\beta)/2}) = O(\Delta^{(5 - 3 \beta) / 2}) \,,
   \end{align*}
   and by the fast zone:
   \begin{align*}
       \text{\# levels} \times \text{\# fast blocks} \times \text{O(\# fast threshold}) + \log n)) = \Delta^{(1-\beta) / 2} \Delta^{(1-\beta) 3/2} O(\Delta^{(1+\beta) / 2} ) = O(\Delta^{(5 - 3 \beta) / 2}) \,.
   \end{align*}
   Combining the two, we find the modified algorithm produces a $O(\Delta^{(5 - 3 \beta) / 2})$ coloring with high probability.
\end{proof}

\subsection{A Randomness-Efficient Robust Algorithm}

\begin{theorem}
\Cref{alg:lowrandom} is an adversarially robust $O(\Delta^3)$ coloring algorithm, which uses $\tO(n)$ bits of space (including random bits used by the algorithm).
\end{theorem}

\begin{proof}
The only step of \Cref{alg:lowrandom} that an adversary could make fail is Line \ref{step:picking-empty-Dcj}.

By \Cref{lem:adv-d3-fail}, this happens with $1/\poly(n)$ probability. Assuming Line \ref{step:picking-empty-Dcj} does not fail, \Cref{lem:adv-d3-correct} proves that the output of the algorithm is a valid $(\Delta+1) \Delta^2$ coloring. Finally, \Cref{lem:adv-d3-space} verifies that \Cref{alg:lowrandom} uses at most $\tO(n)$ bits of space and of randomness.
\end{proof}

\begin{algorithm}[!htb]
  \caption{Randomness-Efficient Adversarially Robust $O(\Delta^{3})$-Coloring in Semi-Streaming Space \label{alg:lowrandom}}
  \begin{algorithmic}[1]
  
  \Statex \textbf{Input}: Stream of edge insertions of an $n$-vertex graph $G=(V,E)$
  \Statex
    \Statex \underline{\textbf{Initialize}:}
    \Statex Define $P := \ceil{10 \log n}$, and let $\ell = 2^{\floor{\log \Delta}}$ be the greatest power of $2$ which is $\le \Delta$
    \Statex Let $\cU$ be a 4-independent family of hash functions from $V$ to $[\ell^2]$, of size $\poly(n)$
    
    \For{ $i \in [\Delta], j \in [P]$} 
      \State $h_{i,j}$ be a uniformly random function from $\cU$ mapping $V$ to $[\ell^2]$
      \State $D_{i,j} \gets \emptyset$  \Comment{Either a set of $h_{i,j}$-monochromatic edges, or $\bot$ after invalidation}
    \EndFor
    \State $B \gets \emptyset$ \Comment{buffer of edges from this epoch}
    \State $\curr \gets 1$ \Comment{current epoch number}
    
  \Statex 
  \Statex \underline{\textbf{Process}(edge $\{u,v\}$):}
  \If{$|B| = n$} \label{step:lowrand-epoch-reset}
    \State  $B \gets \emptyset$; $\curr \gets \curr + 1$ \Comment{End current epoch, switch to next}
  \EndIf
  \State $B \gets B \cup \{\{u,v\}\}$; \Comment{Update current buffer}
  \For{$i$ from $\curr+1$ to $\Delta$, and $j \in [P]$}
    \If{$h_{i,j}(u) = h_{i,j}(v)$} \Comment{For $h_{i,j}$-monochromatic edges...}
      \If{$D_{i,j} \ne \bot \land |D_{i,j}| < \frac{7 n}{\Delta}$} \label{step:wiping-overfull-Dcj-cond}
        \State $D_{i,j} \gets D_{i,j} \cup \{\{u,v\}\}$ \Comment{Record edge in $D_{i,j}$ if there is space}
      \Else
        \State $D_{i,j} \gets \bot$ \label{step:wiping-overfull-Dcj-action} \Comment{Wipe buffer $D_{i,j}$ if it gets too large}
      \EndIf
    \EndIf
  \EndFor
      
  \Statex
  \Statex \underline{\textbf{Query}():}
  \State Let $k = \min \{ j \in [P] : D_{\curr, j} \ne \bot \}$ \label{step:picking-empty-Dcj}  \Comment{This can fail if all $D_{\curr, j} = \bot$}
  \State Let $\chi = $ greedy coloring of $D_{\curr, k} \cup B$
  \State Output the coloring where $y \in V$ is assigned $(\chi(y), h_{\curr,j}(y)) \in [(\Delta+1)] \times [\ell^2]$

  
      
  \end{algorithmic}
\end{algorithm}

\begin{lemma}\label{lem:adv-d3-fail}
  Line \ref{step:picking-empty-Dcj} of \Cref{alg:lowrandom} will execute successfully, with high probability, on input streams provided by an adaptive adversary.
\end{lemma}

\begin{proof}
  We first remark that the time range in which \Cref{alg:lowrandom} updates a given set $D_{i,j}$ is disjoint from and happens before \Cref{alg:lowrandom} first uses the set $D_{i,j}$. The set $D_{i,j}$ is only updated when $\curr < i$; and only used in the query routine when $\curr = i$. Consequently, looking at the outputs of the algorithm does not help an adversary ensure any property of $D_{i,j}$. It suffices, then, to prove that for a given $i$, that Line \ref{step:picking-empty-Dcj} succeeds with high probability on any fixed input stream.
  
  Let $G$ be the graph encoded by the first $n (i - 1)$ edges of the input stream. We will prove that for each $j \in [P]$,
  \begin{align}
    \Pr\left[D_{i,j} \ge \frac{7 n}{\Delta}\right] \le \frac{1}{2} \,. \label{eq:adv-d3-fail-claim}
  \end{align}
  Since the $h_{i,j}$ are chosen independently, the event from Eq. \ref{eq:adv-d3-fail-claim} is true for all values of $j \in P$ is $\le (1/2)^P \le 1/n^{10}$; thus Line \ref{step:picking-empty-Dcj} succeeds with high probability.
  
  Now fix $j$; for each $v \in V$, and $b \in [\ell^2]$, let $X_{v,b}$ be the indicator random variable which is $1$ if $h_{i,j}(v) = b$. We have
  \begin{align*}
    |D_{i,j}| = \sum_{\{u,v\} \in G} \sum_{b \in [\ell^2]} X_{u,b} X_{v,b} \,.
  \end{align*}
  Because $h_{i,j}$ is drawn from a 4-independent family, in particular we have $\Pr[h_{i,j}(u) = h_{i,j}(v) = b] = 1/\ell^4$, so
  \begin{align*}
    \EE |D_{i,j}| = \sum_{\{u,v\} \in G} \frac{1}{\ell^2} = \frac{ |G|}{\ell^2} \le \frac{4 |G|}{\Delta^2} \,,
  \end{align*}
  and, letting $P_3(G) = \{ \{u,v,w\} \in V^3 : \{u,v\} \in G \land \{v,w\} \in G \}$ be the set of $\le |G|\Delta$ paths of length 2,
  \begin{align*}
    \Var |D_{i,j}| &= \EE |D_{i,j}|^2 - (\EE |D_{i,j}|)^2  \\
        & =  \EE \left(\sum_{\{u,v\} \in G} \sum_{b \in [\ell^2]} X_{u,b} X_{v,b}\right)^2 - (\EE |D_{i,j}|)^2 \\
        & \le \sum_{\{u,v\} \in G} \sum_{ \{v ,y\} \in G : \{u,v\} \cap \{v,y\} = \emptyset } \sum_{b \in [\ell^2]} \sum_{c \in [\ell]^2} \EE X_{u,b} X_{v,b} X_{w,c} X_{y,c} \\
        & \qquad + \sum_{\{u,v,w\} \in P_3(G)} \sum_{b \in [\ell^2]} \EE[X_{u,b} X_{v,b}X_{w,b}] \\
        & \qquad  + \sum_{\{u,v\} \in P_3(G)} \sum_{b \in [\ell^2]} \EE[X_{u,b} X_{v,b}] - (\EE |D_{i,j}|)^2 \,.
  \end{align*}
  By the 4-independence of the family from which $h_{i,j}$ is drawn, we have $\EE[X_{u,b} X_{v,b} X_{w,c} X_{y,c}] = 1/\ell^8$, $\EE[X_{u,b} X_{v,b} X_{w,b}] = 1/\ell^6$ and $\EE[X_{u,b} X_{v,b}] = 1 / \ell^4$, so:
  \begin{align*}
     \Var |D_{i,j}| \le \frac{|G|^2}{\ell^4} - \left(\frac{|G|}{\ell^2}\right)^2 + \frac{|G|\Delta}{\ell^4} + \frac{|G|}{\ell^2} \le \frac{16 |G|}{\Delta^3} + \frac{4 |G|}{\Delta^2} \,.
  \end{align*}
  Because a graph of maximum degree $\Delta$ can contain at most $\frac{n \Delta}{2}$ edges, $|G| \le \frac{n \Delta}{2}$, so:
  \begin{align*}
      \EE |D_{i,j}| \le \frac{2 n}{\Delta} \qquad \text{and} \qquad \Var |D_{i,j}| \le \frac{10 n}{\Delta} \,.
  \end{align*}
  By Chebyshev's inequality:
  \begin{align*}
    \Pr\left[|D_i,j| \ge \frac{7 n}{\Delta}\right] \le \Pr\left[||D_i,j| - \EE|D_i,j| | \ge \frac{5 n}{\Delta}\right] \le \frac{(10 n)/\Delta}{((5 n)/\Delta)^2} \le \frac{10}{25} \frac{\Delta}{n} \le \frac{1}{2} \,.
  \end{align*}
  This is precisely Eq. \ref{eq:adv-d3-fail-claim}.
\end{proof}

\begin{lemma}\label{lem:adv-d3-correct}
  If Line \ref{step:picking-empty-Dcj} does not fail, then \Cref{alg:lowrandom}
  outputs a valid $(\Delta+1)(\Delta^2)$ coloring of the input graph.
\end{lemma}

\begin{proof}
  We need to prove that for each edge $\{u,v\}$ in the graph, the coloring assigns
  different values to $u$ and to $v$. Let $k$ be the value of $k$ chosen at Line \ref{step:picking-empty-Dcj}, and let $c$ be the current value of $\curr$. Since
  $D_{\curr,k} \ne \bot$, the set $D_{c,k}$ contains all edges $\{a,b\}$ in 
  the graph for which $h_{c,k}(a) = h_{c,k}(b)$, and, at the time the edge was added, $\curr < c$. All edges for which $\curr = c$ held at the time the edge was added are stored in 
  $B$. If $h_{\curr,k}(u) \ne h_{\curr,k}(v)$, then the colors $(\chi(u),h_{\curr,k}(u))$
  and  $(\chi(v),h_{\curr,k}(v))$ assigned to $u$ and $v$ differ in the second coordinate.
  Otherwise, the edge $\{u,v\} \in D_{\curr,k} \cup B$, so the greedy coloring of 
  $D_{\curr,k} \cup B$ will assign different values to $\chi(u)$ and $\chi(v)$. This ensures the colors assigned to $u$ and $v$ differ in the first coordinate.
  
  Finally, the output color space $[(\Delta+1)] \times [\ell^2]$ has size $\le (\Delta+1) \Delta^2 = O(\Delta^3)$.
\end{proof}

\begin{lemma}\label{lem:adv-d3-space}
  \Cref{alg:lowrandom} requires only $\tO(n)$ bits of space; this includes
  random bits.
\end{lemma}

\begin{proof}
  Because $|\cU| = O(\poly n)$, picking a random hash function from $\cU$ requires only $O(\log n)$ random bits. As the algorithm stores $\Delta P = O(\Delta \log n)$ of these hash functions as $(h_{i,j})_{i \in [\Delta], j \in [P]}$, the total space needed by these function is $O(\Delta (\log n)^2)$.
  
  Next, for each of the sets of edges $D_{i,j}$, for $i \in [\Delta], j \in [P]$, Lines \ref{step:wiping-overfull-Dcj-cond} through \ref{step:wiping-overfull-Dcj-action} ensure that $|D_{i,j}|$ is always $\le \frac{7 n}{\Delta} + 1$; sets that grow too large are replaced by $\bot$. Since edges can be stored using $O(\log n)$ bits, the total space
  usage of all the $D_{i,j}$ is $O\left(\frac{n}{\Delta} \right) \Delta P \cdot O(\log n) = O\left(n (\log n)^2\right)$. Similarly, the buffer $B$ never contains more than $n$ edges, since it is reset when the condition of Line \ref{step:lowrand-epoch-reset} is true; thus $B$ can be stored with $O(n \log n)$ bits. The counter $\curr$ is negligible.
  
  In total, the algorithm needs $O\left(\Delta (\log n)^2\right) + O\left(n (\log n)^2\right) = \tO(n)$ bits of space.
\end{proof}

%% file: tech-proofs.tex

\section{Deferred Proofs of Technical Lemmas}\label{sec:deferred-proofs}

\begin{lemma}[Restatement of \Cref{lem:find-iset}]
    Given a graph $G$ with $m$ edges and $n$ vertices, one can find an independent set of size $\ge n^2 / (2m + n)$ in deterministic polynomial time.
\end{lemma}

\begin{proof}
  We prove that we can in deterministic polynomial time find an independent set in graph $G$ of size $\ge \psi(G) := \sum_{x \in V} \frac{1}{\deg x + 1}$. By Jensen's inequality,
  \begin{align*}
     \psi(G) \ge \frac{|V|^2}{\sum_{x \in V} (\deg x + 1) } = \frac{n^2}{n + 2m}
  \end{align*}
  This is better than required for this lemma.
  
  The procedure is straightforward: let $U \gets V$ be the set of "uncovered" vertices, and $I \gets \emptyset$ the independent set, which we will progressively expand. While $U$ is not empty, pick $x \in U$ minimizing $\sum_{y \in N[x]} \frac{1}{\deg_{G[U]} (y) + 1}$, and remove the closed neighborhood $N[x]$ from $U$, and add $x$ to $I$.
  To prove that this produces a set $I$ of size $\ge \phi(G)$, we show that every time a new vertex is picked, $\phi(G[U])$ decreases by at most 1. First, note that:
  \begin{align*}
      \min_{x \in U} \sum_{y \in N[x]} \frac{1}{\deg_{G[U]} (y) + 1} \le \frac{1}{|U|} \sum_{x \in U} \sum_{y \in N[x]} \frac{1}{\deg_{G[U]} (y) + 1} = \frac{1}{|U|} \sum_{z \in U} \frac{|N[z]|}{\deg_{G[U]} (z) + 1} = \frac{|U|}{|U|} = 1
  \end{align*}
  Second, 
  \begin{align*}
    \phi(G[U]) - \phi(G[U \setm N[x]]) &= \sum_{z \in U} \frac{1}{\deg_{G[U\setm N[x]]} z + 1} - \sum_{z \in U\setm N[x]} \frac{1}{\deg_{G[U\setm N[x]]} z + 1} \\
        & = \sum_{z \in N[x]} \frac{1}{\deg_{G[U\setm N[x]]} z + 1} + \sum_{z \in U \setm N[x]} \left(\frac{1}{\deg_{G[U]}(z)+1} - \frac{1}{\deg_{G[U\setm N[x]}(z) + 1}\right) \\
        &\le \sum_{z \in U} \frac{1}{\deg_{G[U\setm N[x]]} z + 1} + \sum_{z \in U \setm N[x]} 0
  \end{align*}
  because $\deg_{G[U]}(z) \ge \deg_{G[U] \setm N[x]}(z)$. Combining these two inequalities gives $\phi(G[U \setm N[x]]) \ge \phi(G[U]) -  1$.

\end{proof}

\begin{lemma}[Restatement of \Cref{lem:forward-concentration}] Let $X_1,\ldots,X_k$ be a series of $\{0,1\}$ random variables, and $c_1,\ldots,c_k$ real numbers for which for all $i \in k$, $\EE[X_i \mid X_1,\ldots,X_{i-1}] \le c_i$. Then:
  \begin{align}
    \Pr\left[\sum_{i \in [k]} X_i \ge (1 + t) k c \right] \le 2^{- t k c} \qquad \text{assuming $t \ge 3$} \label{eq:fwd-conc-res}
  \end{align}
\end{lemma}

\begin{proof}
  This mostly repeats the proof of the Chernoff bound, albeit using bounds on the conditional expectations instead of independence. First, note that for any $s > 1$, $i \in [k]$, because $\EE[X_i \mid X_1,\ldots,X_{i-1}] \le c$, we also have  $\EE[e^{s X_i} \mid X_1,\ldots,X_{i-1}] \le c (e^s - 1) + 1 \le e^{c (e^s - 1)}$. Then with $s = \ln(1 + t)$, 
  \begin{align*}
    \Pr[\sum_{i \in [k]} X_i \ge c k (1 + t)] &= \Pr[e^{s \sum_{i \in [k]} X_i} \ge e^{s c k (1 + t)}] \\
        &\le e^{-s c k (1 + t)} \EE\left[ e^{s \sum_{i \in [k]} X_i} \right] &&\text{by Markov} \\
        & = e^{-s c k (1 + t)} \EE\left[e^{s X_1} \EE\left[e^{s X_2} \cdots \EE[ e^{s X_k} \mid X_1,\ldots,X_{k-1} ] \cdots \mid X_1 \right] \right] \\
        & \le  e^{-s c k (1 + t)} (e^{c (e^s - 1)})^k \\
        & = \left( \frac{e^t}{(1+t)^{1+t}} \right)^{c k} \,.
  \end{align*}
  For all $t \ge 3$, we have $(1+t) \ln(1+t) \ge (1 + \ln(2)) t$, so:
  \begin{align*}
     \left( \frac{e^t}{(1+t)^{1+t}} \right) = e^{t - (1+t) \ln(1+t)} \le e^{-t \ln 2} = 2^{-t} \,,
  \end{align*}
  which implies Eq. \ref{eq:fwd-conc-res}.
\end{proof}

\begin{lemma}[Restatement of \Cref{lem:gw-property}]
For $p \ge 8 n \log n$, and $\bw = (w_{x,\bj})_{x \in U, \bj \in \b^k}$ there is a function $g_{\bw} \colon U \times [p] \to \b^k$ satisfying:
    \begin{align*}
          \frac{|g_{\bw}^{-1}(x,\bj)|}{p} \le 
          w_{x,\bj} \left(1 + \frac{1}{8\log n}\right) \,, \quad
          \forall~ \bj \in \b^k
    \end{align*}
\end{lemma}

\begin{proof}
As
$\sum_{\bj\in\b^k} w_{x,\bj} = 1$, we can do this by directing the first
$\floor{p w_{x,\bzero} (1 + 1/(8 \log n))}$ entries of $g_{\bw}(x, \cdot)$ to the
pattern $\bzero$; the next $\floor{p w_{x,\bone} (1 + 1/(8 \log n))}$ entries
to the pattern $\bone$; and so on (where $\bzero, \bone, \ldots$ is an
enumeration of $\b^k$), stopping as soon as all $p$ entries of $g_{\bw}(x, \cdot)$
are filled. 

We now argue that $g_{\bw}$ is well-defined, i.e., that every entry $g_{\bw}(x, \cdot)$ is
indeed filled. Examining \cref{eq:slack-def}, since every slack value is at most
$n$, every nonzero value $w_{x,\bj}$ is $\ge 1/n$. Recalling that $p \ge 8n
\log n$, we observe that for such $\bj$,
\begin{align*}
  \floor{p w_{x,\bj} \left(1+\frac{1}{8 \log n}\right)} 
  \ge p w_{x,\bj} + \frac{p w_{x,\bj}}{8 \log n} - 1
  \ge p w_{x,\bj} + \frac{(8n \log n) (1/n)}{8 \log n} - 1 
  = p w_{x,\bj} \,,
\end{align*}
so a total of $\ge \sum_{\bj \in \b^k} p w_{x,\bj} \ge p$ entries $g_{\bw}(x,
\cdot)$ will be covered.
\end{proof}

Finally, we provide the promised formal proof of a key claim made within our
proof of \Cref{lem:c-degree-limit}. We continue to use the notation and
terminology from that proof.

\begin{lemma}[Key claim in proof of \Cref{lem:c-degree-limit}]\label{lem:c-degree-limit-formal} 
    That $\EE[Z_{k,\ell} \mid Z_{\prec (k,\ell)}] \le 1/\Delta^{3/2}$.
\end{lemma}
\begin{proof}
  To express this more formally, we first apply the law of total probability, and expand the definition of $Z_{k,\ell}$:
  \begin{align}
      \EE[Z_{k,\ell} \mid& Z_{\prec (k,\ell)}] \\
      &= \Pr\left[ g_\ell(Y_k) = g_\ell(x) \land \ell \ge \ceil{\frac{\max(d(x),d(Y_k))}{\sqrt{\Delta}}} + 
            1 \land k \le D ~\Big\vert~ Z_{\prec (\ell,\ell)} \right] \nonumber \\
      &= \sum_{v \in V\setm\{x\}} \sum_{c \in [\Delta]^{3/2}} 
            \Pr\left[g_\ell(v) = c \land Y_k = v \land g_\ell(x) = c \land \ell \ge \ceil{\frac{\max(d(x),d(Y_k))}{\sqrt{\Delta}}} + 
            1 \land k \le D ~\Big\vert~ Z_{\prec (\ell,\ell)} \right] \nonumber \\
      &= \sum_{v \in V\setm\{x\}} \sum_{c \in [\Delta]^{3/2}} 
            \Pr\left[g_\ell(v) = c \mid E_{\ell,k,v,c}, Z_{\prec (k,\ell)} \right] 
            \Pr[E_{\ell,k,v,c} \mid Z_{\prec (k,\ell)} ] \,. \label{eq:cond-fast-ub-semiexp}
  \end{align}
  The last step abbreviates the event $\{Y_k = v \land g_\ell(x) = c \land \ell \ge \ceil{\frac{\max(d(x),d(Y_k))}{\sqrt{\Delta}}} + 1 \land k \le D \} =: E_{\ell,k,v,c}$. (In plain terms, this event occurs if it is true that "whether $\{x,Y_k\}$ is stored is determined by the check $g_\ell(v) \overset{?}{=} c$" .) We will now prove that $\Pr\left[g_\ell(v) = c \mid E_{\ell,k,v,c}, Z_{\prec (k,\ell)} \right] = \Pr[g_\ell(v) = c]$ -- in other words, that the event $\{g_\ell(v) = c\}$ is mutually independent of the event $E_{\ell,k,v,c}$ and the random variable $Z_{\prec (k,\ell)}$. This will be done in two steps: first we will show that conditioned on the event $E_{\ell,k,v,c}$ being true, $\{g_\ell(v) = c\}$ and $Z_{\prec (k,\ell)}$ are independent of each other. Then we will prove $\{g_\ell(v) = c\}$ is independent of whether the event $E_{\ell,k,v,c}$ holds.
  
  If $E_{\ell,k,v,c}$ holds, then by definition we have $Y_k = v$. Because the endpoints of the edges $\{x,Y_2\}$, ... $\{x,Y_k \}$ are disjoint, this ensures that $x$, $v$, and $Y_1$ through $Y_k$ are all distinct; consequently $g_\ell(x)$, $g_\ell(v)$,
  $g_\ell(Y_1)$, through $g_\ell(Y_{k-1})$, and  $g_\ell(v)$ are all mutually independent of each other, as are all the functions $g_1, g_2,\ldots,g_{\sqrt{\Delta}}$. Next, because $E_{\ell,k,v,c}$ implies $\ell \ge \ceil{\frac{\max(d(x),d(v))}{\sqrt{\Delta}}} + 1$, we observe that the value of $g_\ell(v)$ has not been revealed to the adversary. According to the code of  \Cref{alg:robust} near Line \ref{step:fast-color}, the value of $g_\ell$ will only be used to produce colorings for vertices $w$ that satisfy $\ceil{d(w) / \sqrt{\Delta}} = \ell$; but $d(v)$ is too low for this to occur. As $g_\ell(v)$ does not affect the output of the algorithm, it also can not affect the behavior of the adversary. Consequently, the sequence $Y_1,\ldots,Y_\ell$, and way in which the degrees of these vertices change, must have been chosen independently of $g_\ell(v)$, conditioning on the event $E_{\ell,k,v,c}$. Because $Z_{\prec (k,\ell)}$ is determined by the algorithm input and the values $g_{\ell'}(x)$, $g_{\ell'}(Y_{k'})$ for all $(k',\ell') \prec (k,\ell)$, and we have shown the latter are mutually independent of $g_\ell(v)$, it follows that $Z_{\prec (k,\ell)}$ is independent of $g_\ell(v)$, conditioned on the event $E_{\ell,k,v,c}$. Thus $\Pr\left[g_\ell(v) = c \mid E_{\ell,k,v,c}, Z_{\prec (k,\ell)} \right] = \Pr[g_\ell(v) = c \mid E_{\ell,k,v,c}]$.
  
  We now prove that $\{g_\ell(v) = c\}$ is independent of the event $E_{\ell,k,v,c}$. We can split $E_{\ell,k,v,c}$ into the intersection of two smaller events; that $\{g_\ell(x) = c\}$, and the event $E_{\ell,k,v} := \{Y_k = v \land \land \ell \ge \ceil{\frac{\max(d(x),d(Y_k))}{\sqrt{\Delta}}} + 1 \land k \le D \}$. Since $v \ne x$, the values of $g_\ell(v)$ and $g_\ell(x)$ are independent; throughout the following argument, we will condition on the event that $g_\ell(x) = c$. The event $E_{\ell,k,v}$ depends only on $Y_k, d(x), d(Y_k)$ and $D$: that is, values derived purely from the input stream the adversary creates, and not otherwise dependent on the random bits of \Cref{alg:robust}. Let $F$ be the event that the value of $g_\ell(v)$ is used to compute a coloring provided to the adversary. If, the event $F$ is does not occur, then the input stream is independent of $g_\ell(v)$, so $E_{\ell,k,v}$ is independent of $\{g_\ell(v) = c\}$. On the other hand, if $F$ does occur, then
  $\ceil{d(v)/\sqrt{\Delta}} = \ell$ must have been true at some point, which
  means the condition $\ell \ge \ceil{\frac{\max(d(x),d(Y_k))}{\sqrt{\Delta}}} + 1$ is false, and $E_{\ell,k,v}$ does not occur. Either way, $g_\ell(v)$ is independent of $E_{\ell,k,v}$. Since this is true no matter whether $g_\ell(x) = c$ holds, it follows that $\Pr[g_\ell(v) = c \mid E_{\ell,k,v,c}] = Pr[g_\ell(v) = c]$.
  
  It remains to finish the upper bound on Eq. \ref{eq:cond-fast-ub-semiexp}. As $Pr[g_\ell(v) = c] = 1/\Delta^{3/2}$,
  \begin{align*}
      \EE[Z_{k,\ell} \mid Z_{\prec (k,\ell)}] &=
         \sum_{v \in V\setm\{x\}} \sum_{c \in [\Delta]^{3/2}} 
            \Pr[g_\ell(v) = c ] \Pr[E_{\ell,k,v,c} \mid Z_{\prec (k,\ell)} ] \\
         &\le 
         \frac{1}{\Delta^{3/2}} \left( \sum_{v \in V\setm\{x\}} \sum_{c \in [\Delta]^{3/2}} 
            \Pr[E_{\ell,k,v,c} \mid Z_{\prec (k,\ell)} ] \right) \\
        &\le 
         \frac{1}{\Delta^{3/2}} \left( \sum_{v \in V\setm\{x\}} \sum_{c \in [\Delta]^{3/2}} 
            \Pr[Y_k = v \land g_\ell(x) = c \mid Z_{\prec (k,\ell)} ] \right) \\
        &= \frac{1}{\Delta^{3/2}} \,. \qedhere
  \end{align*}
\end{proof}